\documentclass[11pt]{article}
\usepackage{microtype}
\usepackage{graphicx} 
\usepackage{fullpage}

\usepackage{mathpazo}
\usepackage{diagbox}

\usepackage[T1]{fontenc}

\usepackage{amsmath,amsfonts,amssymb}
\usepackage{mathtools}
\usepackage{amsthm} 

\usepackage{graphicx}
\usepackage{subcaption} 

\usepackage{circuitikz}
\usepackage{tikz}
\usetikzlibrary{calc}

\usepackage{physics}
\usepackage{bm} 

\usepackage{enumitem}
\usepackage{multicol}

\usepackage{bbm}
\usepackage{dsfont}

\usepackage{nicefrac}
\usepackage{float} 

\usepackage{todonotes}
\usepackage{mdframed}
\usepackage{longfbox}
\usepackage{lipsum}
\definecolor{NotSoDarkBlue}{RGB}{15,15,210}
\usepackage[pdfstartview=FitH,pdfpagemode=UseNone,colorlinks,linkcolor=NotSoDarkBlue,filecolor=blue,citecolor=NotSoDarkBlue,urlcolor=NotSoDarkBlue,breaklinks]{hyperref}
\usepackage{cleveref}

\usepackage{longtable}
\usepackage{array}

\newcommand{\N}{\mathbb{N}}

\newcommand{\R}{\mathbb{R}}

\newcommand{\C}{\mathbb{C}}
\newcommand{\F}{\mathbb{F}}

\newtheorem{theorem}{Theorem}[section]
\newtheorem{proposition}[theorem]{Proposition}
\newtheorem{definition}[theorem]{Definition}
\newtheorem{lemma}[theorem]{Lemma}

\newtheorem{remark}[theorem]{Remark}
\newtheorem{corollary}[theorem]{Corollary}
\newtheorem{fact}[theorem]{Fact}

\newcommand{\secparam}{n}

\renewcommand{\bra}[1]{\langle#1\rvert}
\renewcommand{\braket}[2]{\langle #1 \mid #2 \rangle}
\renewcommand{\ket}[1]{\lvert#1\rangle}

\newcommand{\cK}{{\mathcal K}}

\newcommand{\bfG}{\mathbf{G}}

\newcommand{\sfD}{\mathsf{D}}
\newcommand{\sfE}{\mathsf{E}}
\newcommand{\sfF}{\mathsf{F}}

\newcommand{\regA}{{\color{gray} \mathbf{A}}}
\newcommand{\regB}{{\color{gray} \mathbf{B}}}
\newcommand{\regC}{{\color{gray} \mathbf{C}}}

\newcommand{\regE}{{\color{gray} \mathbf{E}}}

\newcommand{\regR}{{\color{gray} \mathbf{R}}}
\newcommand{\regS}{{\color{gray} \mathbf{S}}}

\allowdisplaybreaks

\newcommand{\keygen}{\mathsf{Gen}}
\newcommand{\enc}{\mathsf{Enc}}
\newcommand{\dec}{\mathsf{Dec}}

\newcommand{\ct}{\mathsf{ct}}
\newcommand{\linear}{{\cal L}}
\newcommand{\density}{{\cal D}}
\newcommand{\dimension}{{\sf dim}}

\newcommand{\alice}{{\cal A}}
\newcommand{\bob}{{\cal B}}
\newcommand{\charlie}{{\cal C}}
\newcommand{\expect}{\mathbb{E}}

\usepackage[backend=biber,style=alphabetic,maxnames=20]{biblatex}
\newcommand{\prob}{{\sf Pr}}
\newcommand{\negl}{\operatorname{negl}}
\newcommand{\qpt}{\mathsf{QPT}}

\title{Unconditional Unclonable Encryption}
\date{July 22, 2026}
\author{Prabhanjan Ananth\footnote{\texttt{prabhanjan@cs.ucsb.edu}}\\ {\small UCSB} \and Seyoon Ragavan\footnote{\texttt{sragavan@mit.edu}}\\{\small MIT} \and Amit Sahai\footnote{\texttt{sahai@cs.ucla.edu}}\\{\small UCLA}} 

\date{}

\begin{document}

\maketitle

\begin{abstract}
\noindent We give an unconditional construction of information-theoretically secure one-time private-key unclonable encryption for one-bit messages, with efficient encryption and decryption and exponentially small unclonable-indistinguishability advantage. Previous work by Bhattacharyya and Culf (Nature Physics, 2026) and Bhattacharyya, Broadbent, and Culf (arXiv:2603.08916) either only showed inverse-polynomial security loss or required inefficient encryption/decryption operations; we avoid both of these caveats.

Our scheme's key is a uniformly random non-identity phase-free Pauli on $n$ qubits, and bit $a$ is encrypted as a random $(-1)^a$ eigenstate of that Pauli. We prove that the probability that both receivers recover the bit is at most $1/2 + O\left(2^{-n/2}\right).$ By a lower bound due to Broadbent, Culf, and Rochette (IACR Communications in Cryptology, 2026), this is the best probability bound achievable with $n$-qubit ciphertexts while ensuring perfect correctness (up to the constant hidden in the $O(\cdot)$).
\end{abstract}

\tableofcontents

\section{Introduction}

\noindent Unclonable cryptography is a branch of quantum cryptography that leverages the no-cloning principle of quantum mechanics~\cite{WoottersZurek82,Dieks82,BroadbentLord20} to design primitives that are impossible to achieve classically. 
No-cloning is merely a hope; it precludes the possibility of an adversary generically creating two copies of a given state, but an unclonable security guarantee must prevent an adversary from creating \emph{any two states} that retain some property of the given state.

Within the rapidly growing landscape of unclonable cryptography, unclonable encryption has emerged as an important primitive. Roughly speaking, unclonable encryption (UE)~\cite{Gottesman03,BroadbentLord20} is an encryption scheme whose ciphertexts cannot be copied in a way that retains information about the plaintext message (even if the decryption key is revealed).  In the one-time setting, an adversary receives a single copy of a quantum ciphertext while the key is hidden and then splits it between two non-communicating recipients.  The key is then revealed to both the recipients, and the security requires that they cannot both simultaneously recover the hidden message substantially better than a coordinated guess. 

In more detail, two notions of security have been studied for unclonable encryption:
\begin{itemize}
    \item {\bf Search security}: the challenger samples a message $m$ uniformly at random from $\{0,1\}^{\secparam}$ and sends an encryption of $m$ to $\alice$ who then produces a bipartite state shared with $\bob$ and $\charlie$. Upon receiving the decryption key, both $\bob$ and $\charlie$ respectively output $m_{\bob}$ and $m_{\charlie}$. $(\alice,\bob,\charlie)$ succeed if $m_{\bob}=m$ and $m_{\charlie}=m$. The work of~\cite{BroadbentLord20} presented a scheme where they showed that the probability of any adversary $(\alice,\bob,\charlie)$ succeeding is  at most $(\cos^{2}(\frac{\pi}{8}))^n$.
    \item {\bf Indistinguishability security}: the challenger samples a bit $b$ at random and then encrypts a message $m_b$, for two adversarially chosen messages $(m_0,m_1)$, and then sends the ciphertext to $\alice$. From here on, the experiment is the same as the search security experiment. Unlike search security, there is always an adversarial strategy that can win with probability $\frac{1}{2}$.
    For example, $\alice$ could forward the ciphertext to $\bob$, who can later decrypt and leave $\charlie$ to guess randomly. Alternatively, $\bob$ and $\charlie$ could just agree to both guess 0 regardless of $\alice$'s message.  
    So the natural goal is to design an UE scheme satisfying indistinguishability security where the success probability is as close to $\frac{1}{2}$ as possible.   
\end{itemize}

\paragraph{Prior constructions.} Over the past few years, many works have studied indistinguishability security for unclonable encryption. We use ${\sf indUE}$ to denote unclonable encryption satisfying indistinguishability security. 

\begin{itemize}
    \item {\bf Plain model}:~\cite{AnanthKaleogluYuen25} achieved ${\sf indUE}$ where the decryption keys are quantum states. Recently, there has been progress in designing ${\sf indUE}$ with classical keys.~\cite{BotteronEtAl26} proposed a construction of ${\sf indUE}$ with $n$-qubit ciphertexts, conjectured an inverse-polynomial advantage bound of $O(1/\sqrt{n})$, proved this bound for $n \leq 3$, and provided numerical evidence for this bound for $4 \leq n \leq 8$. This was improved by~\cite{BhattacharyyaCulf26}, who proposed a different construction of ${\sf indUE}$ for which they proved an inverse-polynomial advantage bound. A follow up work by~\cite{BhattacharyyaBroadbentCulf26} presents a construction with the desired negligible advantage bound; however, their construction is not efficient, meaning that the encryption and the decryption algorithms do not run in quantum polynomial time\footnote{Concretely, their scheme has advantage negligible in $\secparam$, where  encryption and decryption can be implemented in time  exponential in $\secparam$.}.
    \item {\bf Idealized models}:~\cite{AnanthEtAl22,AnanthKaleogluLiu23} presented  ${\sf indUE}$ with optimal advantage in the quantum random oracle model.~\cite{BartusekGoldin2026} also presented optimal ${\sf indUE}$ with stronger security guarantees albeit in the Haar random oracle model.
\end{itemize}
\noindent Despite the recent progress, the goal of achieving efficient ${\sf indUE}$ in the plain model with classical keys and negligible adversary advantage has remained elusive. 
Although ${\sf indUE}$ could plausibly exist information-theoretically, we do not even have a construction that leans on plain-model cryptographic assumptions. 
Moreover, a few works present evidence highlighting the challenges in achieving indistinguishability security. Notably~\cite{MajenzSchaffnerTahmasbi21,AnanthEtAl22} showed that ${\sf indUE}$ with deterministic encryption cannot satisfy negligible security.~\cite{ColadangeloLiuXie26} also  highlighted the difficulties in using BB84 states to achieve ${\sf indUE}$.  

\paragraph{Our contribution.}
We address the main goal of achieving unclonable encryption satisfying indistinguishability security. We show: 

\begin{theorem}[Informal]
For every $n\geq 1$, there exists a one-time private-key quantum encryption scheme for one-bit messages with a $2n$-bit classical key and an $n$-qubit ciphertext.  The scheme has perfect correctness, uses only single-qubit Clifford gates for encryption and local Pauli measurements for decryption, and satisfies
\[
\prob[b_{\mathcal B}=b_{\mathcal C}=b]
\leq
\frac12+\frac12\sqrt{\frac{2^n}{4^n-1}}
\]
against every information-theoretic pre-key splitting attack in the one-bit unclonable-indistinguishability experiment.
\end{theorem}

\noindent By a lower bound of Broadbent, Culf, and Rochette~\cite[Corollary~6]{broadbent2025optimaluntelegraphableencryptionimplications}, this is the best bound achievable with $n$-qubit ciphertexts and perfect correctness, up to the constant factor in the $2^{-n/2}$ advantage.

\paragraph{Implications.} Quantum state indistinguishability obfuscation ($\mathsf{qsiO}$)~\cite{ColadangeloGunn24} takes a quantum implementation of a classical function---a quantum state together with an evaluation circuit---and produces another implementation of the same function. Its security guarantee is that obfuscations of any two implementations of the same size and functionality are computationally indistinguishable, even to quantum adversaries.

\noindent As a consequence, our unconditional construction of $\mathsf{indUE}$ yields the following applications of $\mathsf{qsiO}$. Both applications assume $\mathsf{qsiO}$ and post-quantum injective one-way functions; our result eliminates the additional assumption that unclonable encryption exists.
\begin{itemize}
    \item \textbf{Copy-protection for search-puncturable programs.}
    Coladangelo and Gunn~\cite[Theorem~18]{ColadangeloGunn24} show that $\mathsf{qsiO}$ copy-protects any family of \emph{search-puncturable} programs, assuming injective one-way functions and unclonable encryption. Informally, such a program can be punctured at an input $x$ while preserving its behavior everywhere else, but the punctured description does not enable an efficient adversary to produce a valid output at $x$. Combining their reduction with our result therefore gives copy-protection for these families from $\mathsf{qsiO}$ and post-quantum injective one-way functions.\footnote{Note that our result is unclonable encryption for single bits, whereas the Coladangelo--Gunn reduction requires unclonable encryption for arbitrary-length messages. This can be generically handled via the length-boosting results by either~\cite[Corollary 53]{AnanthBehera24} or~\cite{DBLP:conf/stoc/MaH25} combined with~\cite{bartusek2026noteboostinguncloneableencryption}. Alternatively, one could use the follow-up by~\cite{BhattacharyyaBroadbentCulfMessages26} on our work which achieves unclonable encryption for arbitrary-length messages without any computational assumptions.} Given one protected program, no efficient adversary can produce two separated evaluators that both return valid outputs on the same uniformly random challenge input, except with negligible probability.

    \item \textbf{Unclonable puncturable obfuscation.}
    Ananth and Behera~\cite[Section~6, Corollary~52]{AnanthBehera24} show that $\mathsf{qsiO}$, together with injective one-way functions and unclonable-indistinguishable encryption of bits, yields unclonable puncturable obfuscation (UPO). Our construction supplies the encryption primitive unconditionally, giving UPO for polynomial-size puncturable keyed circuits from $\mathsf{qsiO}$ and post-quantum injective one-way functions. The guarantee is generalized UPO security for \emph{identical, uniformly random challenges}: after an adversary splits an obfuscation between two non-communicating parties, both parties receive the same random puncture point and must distinguish the original circuit from a circuit whose behavior is modified at that point. Their joint success probability is at most $1/2+\operatorname{negl}(\lambda)$.
\end{itemize}\paragraph{Subsequent work.}
Following the initial version of this work, Bhattacharyya, Broadbent, and Culf~\cite{BhattacharyyaBroadbentCulfMessages26} extended efficient, information-theoretically secure unclonable encryption from a single bit to messages of arbitrary length. Their construction exploits the Clifford structure of the encoding and achieves statistical security in the one-time setting, with encryption and decryption polynomial in the message length and security parameter. Thus, the extension to longer messages can also be achieved without computational assumptions.

Ananth~\cite{AnanthNonLocal26} studied a general non-local search-to-decision question over $\mathbb{F}_2$. For a bipartite encoding of a uniformly random string, he showed that a non-negligible advantage in simultaneously predicting its inner product with a common random mask implies a non-negligible probability of simultaneously recovering the entire string by local measurements. This result is information-theoretic and does not provide an efficient extractor. Coladangelo, Liu, and Xie~\cite{ColadangeloLiuXieGL26} established a simultaneous Goldreich--Levin reduction for two entangled parties receiving the same random mask. Their reduction generically upgrades search security to unclonable indistinguishability and, in particular, proves security of the natural BB84-based construction with a random inner-product mask. They also give an efficient implementation when a lower bound on the prediction advantage is known. 

\paragraph{Concurrent work.}
Concurrently and independently, Botteron, Designolle, and Fawzi~\cite{BotteronDesignolleFawzi26} studied a general family of \emph{Pauli encodings}, in which ciphertexts are normalized eigenspace projectors of Pauli strings. Their framework includes the family of all non-identity tensor Paulis considered here. Their work provides a complementary analysis of the possibilities and limitations of Pauli encodings.

\paragraph{Lean 4 formalization.}
To help with verifiability, we release Lean 4 code for our main proofs.
In more detail, the Lean repository proves all correctness and security claims in Theorem~\ref{thm:onebitsecurity}, but it does not formally show that encryption and decryption are efficient.
It also provides a modular interface defining unclonable encryption for a bit, allowing for other users to provide their own constructions and security proofs in Lean 4.
It also does not prove the aforementioned implications for other unclonable cryptography primitives since that would require formalizing the results of~\cite{AnanthBehera24} and~\cite{ColadangeloGunn24}.
The Lean code was built using Codex GPT-5.6 Sol and GPT-6 Astra and can be found at \href{https://github.com/sragavan99/lean-unclonable-enc-paulis}{this GitHub repository.}

\subsection{Technical Overview}\label{sec:techoverview}

\paragraph{From BB84 parity to Pauli parity.}
It is useful to start with the natural BB84-based attempt. For each qubit,
the encryptor chooses either the $X$ or the $Z$ observable and prepares one
of its two eigenstates. The eigenvalue bits are sampled uniformly subject
to their XOR being the message. Once the basis string is revealed, the
decryptor measures each qubit and takes the parity of the outcomes.
The surprising XOR-repetition result of Coladangelo, Liu, and Xie~\cite{ColadangeloLiuXie26} exhibits an attack on this construction that has constant winning probability $\cos^2(\pi/8)$.

Our construction makes a very simple tweak on this: instead of choosing between $X$ and $Z$ for each qubit, we choose among the full set of Paulis $\{I, X, Y, Z\}$.
In more detail, the key is chosen uniformly
from all $4^n-1$ non-identity Paulis
\[
P_{x,z}=\bigotimes_{j=1}^n P_{x_j,z_j},
\qquad
P_{x_j,z_j}=i^{x_jz_j}X^{x_j}Z^{z_j}\in\{I,X,Y,Z\}.
\]
To encrypt a bit $m$, prepare a uniformly random product eigenstate whose
eigenvalue under $P_{x,z}$ is $(-1)^m$. 
Concretely, choose the eigenvalue
signs on the non-identity factors uniformly subject to their product being
$(-1)^m$, and prepare a maximally mixed qubit on every identity factor.
The resulting ciphertext is the maximally mixed state on the
$(-1)^m$ eigenspace:
\[
\rho_{\ct}^{(k,m)}=\frac{I+(-1)^mP_{x,z}}{2^n}.
\]
Decryption measures the local Paulis and multiplies the signs. The product
is exactly $(-1)^m$. Thus both encryption and decryption use only local
operations; the challenge is to prove security.
\par Our construction is quite close to that conceived by~\cite{BotteronEtAl26}; they take their key space to be a collection of pairwise anticommuting Paulis, whereas we take the full collection of non-identity Paulis. We show that this modification of the construction of~\cite{BotteronEtAl26} satisfies indistinguishability security with negligible advantage. 
A modification abandoning pairwise anticommutation is necessary, as the pairwise-anticommuting constraint imposed by~\cite{BotteronEtAl26} automatically limits the number of keys to at most $\mathsf{poly}(n)$.

Before continuing, we briefly foreshadow that our construction does not actually rely on the key being sampled from the full set of $4^n-1$ non-identity Paulis. 
We could take any nonempty subset $\cK$ of the non-identity Paulis, and our techniques will show a security bound of $\frac{1}{2} + \frac{1}{2} \sqrt{\frac{2^n}{|\cK|}}$.
Thus any key space of size $2^{n+\omega(\log n)}$ would suffice for security; the BB84 construction falls short of this because it comprises exactly $2^n$ keys.

\paragraph{The operators and their meaning.}
The Choi--Jamio\l{}kowski representation lets us describe the cloner's
channel by a fixed state $\tau$ on a reference register $\regA$ and the
receivers' registers $\regB,\regC$. This is the usual ricochet reduction
used in unclonable encryption~\cite{BroadbentLord20}.
For each key $k=(x,z)$, let $B_k$ and $C_k$ be Bob's and Charlie's decoding
observables: outcome $0$ has sign $+1$, and outcome $1$ has sign $-1$. 
By a standard Naimark dilation argument, we may assume without loss of generality that $B_k$ and $C_k$ are unitary reflections.
Writing $K=4^n-1$ for the number of keys, define
\begin{align*}
\mathsf E_B&=\frac1K\sum_k P_{x,z}^T\otimes B_k\otimes I_{\regC},\\
\mathsf E_C&=\frac1K\sum_k P_{x,z}^T\otimes I_{\regB}\otimes C_k,\\
\mathsf E_{BC}&=\frac1K\sum_k I_{\regA}\otimes B_k\otimes C_k,
\qquad
\sfD=I_{\regA\regB\regC}-\mathsf E_{BC}.
\end{align*}
These operators have simple interpretations:
\begin{align*}
\Tr(\mathsf E_B\tau)&=2\prob[\text{Bob guesses correctly}]-1,\\
\Tr(\mathsf E_C\tau)&=2\prob[\text{Charlie guesses correctly}]-1,\\
\Tr(\sfD\tau)&=2\prob[\text{Bob and Charlie disagree}].
\end{align*}
In particular, $\sfD$ measures disagreement, with a factor of two in our
normalization.

Expanding the success probability gives
\begin{equation}\label{eqn:techoverview:winprob}
\prob[\mathsf{Win}]
=\frac12\left(1+\Tr(\bfG\tau)\right)
\leq\frac12\left(1+\lambda_{\max}(\bfG)\right),
\qquad
\bfG=\frac12(\mathsf E_B+\mathsf E_C-\sfD).
\end{equation}
The subtraction of $\sfD$ has a concrete meaning: the individual success
probabilities count cases where only one receiver is correct, and in
exactly those cases the receivers disagree. Our task is to show that the
largest eigenvalue of $\bfG$ is small. We do not need to bound its negative
eigenvalues: for example, if Bob always answers $0$ and Charlie always
answers $1$, then $\bfG=-I$ and their joint success probability is zero.

\paragraph{Why the bounds of~\cite{tomamichel2013monogamy} and SoS techniques appear insufficient.}
Up to this point, the manipulations to analyze the UE security game have been standard~\cite{BroadbentLord20}, modulo our particular choice of UE construction. At this point, we have reduced our task to bounding the top eigenvalue of some Hermitian operator; the two main pre-existing techniques for proceeding from here are a bound by~\cite{tomamichel2013monogamy} and SoS/SDP techniques~\cite{navascues2008convergent} which were recently leveraged by~\cite{BotteronEtAl26}. We therefore take this opportunity to highlight the limitations of these pre-existing approaches for our setting, before continuing with the overview of our proof.

First, the bound by~\cite{tomamichel2013monogamy} is central to their analysis of the BB84 monogamy-of-entanglement game and the construction by~\cite{BroadbentLord20} of search-secure unclonable encryption. 
They show, for any Hermitian projectors $\Pi_1, \ldots, \Pi_K$ acting on the same Hilbert space, that:
\begin{equation}\label{eq:tfkw}
    \norm{\frac{1}{K} \sum_{i = 1}^K \Pi_i}_\infty \leq \frac{1}{K} \sum_{j \in [K]} \max_{i \in [K]} \norm{\Pi_i \Pi_{\pi_j(i)}}_\infty,
\end{equation}
where $\pi_1, \ldots, \pi_K$ are any permutations on $[K]$ such that $\pi_j(i) \neq \pi_{j'}(i)$ whenever $j \neq j'$. To cast our goal as bounding the operator norm of a sum of projectors, note that we can write:
\begin{equation}
    \frac{1}{2}(I+\bfG) = \frac{1}{K} \sum_{k = (x, z) \in \cK} \sum_{b \in \{0, 1\}} \left(\frac{I+(-1)^b P_{x, z}^T}{2}\right) \otimes \left(\frac{I+(-1)^b B_k}{2}\right) \otimes \left(\frac{I+(-1)^b C_k}{2}\right).
\end{equation}
Any attempt to bound the norm of this above by $1/2+\negl(n)$ using~\eqref{eq:tfkw} will be ``bottlenecked'' by the norms of products of Pauli projectors; if we were to take $B_k = C_k = I$ for all $k$, then we have:
\begin{equation}
    \frac{1}{2}(I+\bfG) = \frac{1}{K} \sum_{k = (x, z) \in \cK} \left(\frac{I+P_{x, z}^T}{2}\right) \otimes I \otimes I.
\end{equation}
Now one can check:
\begin{align*}
    \norm{\left(\frac{I+P_{x, z}^T}{2}\right) \left(\frac{I+P_{x', z'}^T}{2}\right)}_\infty &= \begin{cases}
        1,\text{ if }P_{x, z}P_{x', z'} = P_{x', z'} P_{x, z}; \\
        \frac{1}{\sqrt{2}},\text{ if }P_{x, z}P_{x', z'} = -P_{x', z'} P_{x, z}.
    \end{cases}
\end{align*}
Either way, this bound is well above $1/2$; in fact, the above illustrates that using~\eqref{eq:tfkw} can only hope to get us to a bound of $1/\sqrt{2}$ (whereas we want $1/2+\negl(n)$). 

A second approach for analyzing unclonable encryption is based on the NPA sum-of-squares hierarchy for nonlocal games~\cite{navascues2008convergent}; this was recently applied by~\cite{BotteronEtAl26} for small values of $n$. Looking ahead, our particular bounding argument does not admit a very direct translation as a constant-degree sum-of-squares proof because it leverages an infinite power series expansion of a matrix inverse (see~\eqref{eq:filterfirsttime} which we will introduce soon, and its Taylor expansion in~\eqref{eq:ftaylorexpansion} in the body). We emphasize that this intuition does not rule out the possibility of a constant-degree SoS analysis of our scheme obtained via a different route.

\paragraph{A bound from the Pauli structure.}
The first estimate is
\[
\|\mathsf E_B\mathsf E_C\|_\infty\leq\frac dK,
\qquad d=2^n.
\]
We give two proofs. One uses the signs that arise when we commute two
Paulis past each other. These signs form a submatrix of a scaled,
column-permuted Hadamard matrix, whose norm controls the sum defining
$\mathsf E_B\mathsf E_C$. The other uses Hilbert--Schmidt orthogonality:
after tracing out the reference register, the terms from distinct Paulis
vanish. Both proofs apply to arbitrary decoding observables.

It is tempting to interpret this product bound as saying that Bob and
Charlie cannot both have much advantage individually. But a product of
operator expectations need not equal the expectation of the product.
Indeed, consider a cloner who flips a fair coin and forwards the
ciphertext to one receiver, leaving the other to guess randomly. Each
receiver succeeds with probability $3/4$, even though their joint success
probability is only $1/2$. Their errors occur on disjoint events. To bound
joint success, we must account for disagreement as well.

\paragraph{Accounting for disagreement.}
We use the positive operator
\begin{equation}\label{eq:filterfirsttime}
    \sfF(t)=2t(2tI+\sfD)^{-1},\qquad t>0,
\end{equation}

where $t$ is a parameter we will select later.
On an eigenvector of $\sfD$ with eigenvalue $a\geq0$, this operator acts
by the scalar $2t/(2t+a)$. It therefore gives more weight to directions
with less disagreement, which are the directions relevant to a successful
attack. The term $2tI$ makes the inverse well-defined even when $\sfD$
has a kernel.

The product bound survives this insertion:
\begin{equation}\label{eq:filteredoverview}
\|\mathsf E_B\sfF(t)\mathsf E_C\|_\infty\leq\frac dK
\qquad\text{for every }t>0.
\end{equation}
To see why, expand the inverse as a convergent geometric series in
$\mathsf E_{BC}$. Each term involves
$\mathsf E_B\mathsf E_{BC}^{\ell}\mathsf E_C$.
Operators on different receivers' registers commute, so we can move the
inserted Charlie operators to the left and the inserted Bob operators to
the right. They are all contractions, and hence each term has norm at
most $\|\mathsf E_B\mathsf E_C\|_\infty$. The coefficients of the series
are positive and sum to one.

\paragraph{Completing the argument.}
The specific choice of $\sfF$ is in order to enable the last step of the proof.
The remaining analytical task is to convert a bound on the spectral norm of a \emph{product} of matrices that looks roughly like $\sfE_B \sfD^{-1} \sfE_C$ to a bound on the top eigenvalue of a \emph{sum} of matrices, $\sfE_B + \sfE_C - \sfD$.
The solution to this is inspired by the scalar identity:
\begin{equation}\label{eq:diffofsquares}
    (x+y)^2 =  (x-y)^2 + 4xy.
\end{equation}
In our setting, the relevant matrix identity is:
\begin{equation}\label{eq:matrixdiffofsquares}
    (\mathsf E_B+\mathsf E_C)\sfF(t)
(\mathsf E_B+\mathsf E_C) = (\mathsf E_B-\mathsf E_C)\sfF(t)
(\mathsf E_B-\mathsf E_C) + 2 \mathsf{E}_B \mathsf{F}(t) \mathsf{E}_C + 2\mathsf{E}_C \mathsf{F}(t) \mathsf{E}_B.
\end{equation}
We will show a bound on the LHS by bounding the RHS. The RHS has three terms, the second and third of which are handled by~\eqref{eq:filteredoverview}. For the first term on the RHS, we show in Proposition~\ref{clm:security:quadraticsandwich} that: 
\begin{equation}\label{eq:quadraticsandwich}
    (\sfE_B - \sfE_C) \sfF(t) (\sfE_B - \sfE_C) \preceq 2t \sfD.
\end{equation}
The intuition for the above is that $\sfE_B - \sfE_C, \sfE_C - \sfE_B \preceq \sfD$; if exactly one of Bob and Charlie succeeds, then they must disagree. 
Plugging this and~\eqref{eq:filteredoverview} into~\eqref{eq:matrixdiffofsquares} implies that:
\begin{equation}\label{eq:intermediateproducttosum}
    (\mathsf E_B+\mathsf E_C)\sfF(t)
(\mathsf E_B+\mathsf E_C) \preceq 2t\mathsf{D} + \frac{4d}{K} I.
\end{equation}
It remains to show that~\eqref{eq:intermediateproducttosum} implies that $\lambda_{\max}(\bfG) \leq \negl(n)$. 
To this end, we will now set $t=\lambda_{\max}(\bfG)$ and assume $t > 0$ (otherwise we are done). Choose
a unit eigenvector $\ket{\psi}$ with $\bfG\ket{\psi}=t\ket{\psi}$.
Rearranging this equation gives
\[
(\mathsf E_B+\mathsf E_C)\ket{\psi}
=(2tI+\sfD)\ket{\psi}.
\]
Thus the inverse in $\sfF(t)$ cancels the operator on the right. It follows
that
\begin{equation}\label{eq:sum}
    \bra{\psi}(\mathsf E_B+\mathsf E_C)\sfF(t)
(\mathsf E_B+\mathsf E_C)\ket{\psi}
=4t^2+2t\bra{\psi}\sfD\ket{\psi}.
\end{equation}
Therefore~\eqref{eq:intermediateproducttosum} tells us that:
\begin{align*}
    4t^2 + 2t\bra{\psi}\sfD\ket{\psi} &\leq \bra{\psi}2t\sfD + \frac{4d}{K}I\ket{\psi} \\
    &= 2t\bra{\psi}\sfD\ket{\psi} + \frac{4d}{K} \\
    \Leftrightarrow t &\leq \sqrt{d/K}.
\end{align*}
Together with~\eqref{eqn:techoverview:winprob}, this bounds the winning
probability by $\frac12+\frac12\sqrt{d/K}
=\frac12+\frac12\sqrt{2^n/(4^n-1)}$, completing the proof.

\section{Preliminaries}
\noindent The statistical security parameter is denoted by $\secparam\in\N$. We write
$x\xleftarrow{\$}S$ for uniform sampling from a finite set $S$. A function
$\mu:\N\rightarrow\R_{\geq0}$ is \emph{negligible} if, for every $c > 0$, there is an $N_c$ such that $\mu(n)<1/n^c$ for every $n\geq N_c$. 

A register $\regR$ is associated with a finite-dimensional Hilbert space, also denoted $\regR$. We
write $\linear(\regR)$ for the linear operators on $\regR$ and $\density(\regR)$ for the density
operators on $\regR$. An operator $M$ is positive/positive semidefinite, written $M\succeq0$, if
$\bra{\psi}M\ket{\psi}\geq0$ for every vector $\ket{\psi}$. For self-adjoint (or Hermitian\footnote{We are only concerned with finite-dimensional Hilbert spaces.}) operators $M,N$, the
Loewner order $M\preceq N$ means that $N-M\succeq0$.

A density operator $\rho\in\density(\regR)$ satisfies $\rho\succeq0$ and $\Tr(\rho)=1$. A pure state
$\ket{\psi}$ is identified with the rank-one density operator $\ketbra{\psi}{\psi}$. If
$\rho_{\regR\regS}$ is a bipartite state, its reduced state on $\regR$ is
$\rho_{\regR}=\Tr_{\regS}(\rho_{\regR\regS})$. All partial traces in this note are unnormalized as
linear maps on operators.

For an operator $M$, the operator norm and trace norm are
\[
\|M\|_{\infty}
:=
\sup_{\|\psi\|=1}\|M\ket{\psi}\|,
\qquad
\|M\|_1
:=
\Tr\sqrt{M^{\dagger}M}.
\]
We will repeatedly use the equivalent characterization
\[
\|M\|_{\infty}
=
\sup_{\substack{\|\xi\|=1\\\|\eta\|=1}}
\left|\bra{\eta}M\ket{\xi}\right|.
\]
The operator norm is submultiplicative:
$\|MN\|_{\infty}\leq\|M\|_{\infty}\|N\|_{\infty}$. If $M=M^{\dagger}$, then
$\|M\|_{\infty}$ is the largest absolute value of an eigenvalue of $M$.
We also let $\lambda_{\max}(M)$ denote the largest eigenvalue of $M$ (without taking absolute values).
Consequently, for every density operator $\rho$,
\begin{equation}\label{eq:maxeig}
    \Tr(M\rho)\leq \lambda_{\max}(M).
\end{equation}

A quantum channel
$\Phi:\linear(\regR)\rightarrow\linear(\regS)$ is a completely positive trace-preserving linear map. We abuse notation and also use $I_{\regA}:\linear(\regA) \rightarrow \linear(\regA)$ to denote the identity channel.  A measurement with outcome
set $\mathcal X$ is a POVM $\{M_x\}_{x\in\mathcal X}$ satisfying
$M_x\succeq0$ and $\sum_xM_x=I$. On input state $\rho$, the Born rule gives
$\prob[x]=\Tr(M_x\rho)$.

A binary POVM $\{M_0,M_1\}$ can equivalently be represented by the self-adjoint contraction
\[
B:=M_0-M_1,
\qquad
-I\preceq B\preceq I.
\]
Conversely,
\[
M_0=\frac{I+B}{2},
\qquad
M_1=\frac{I-B}{2}.
\]
If outcome $0$ is assigned value $+1$ and outcome $1$ is assigned value $-1$, then
$\Tr(B\rho)=\prob[0]-\prob[1]$ is the signed bias of the outcome.

We also use the following standard facts:
\begin{enumerate}[label=(\roman*)]
    \item every self-adjoint operator has an orthonormal eigenbasis;
    \item if $M\succeq rI$ for some $r>0$, then $M$ is invertible and $0\prec M^{-1}\preceq r^{-1}I$;
    \item if two positive operators commute, then their product is positive;
    \item $\Tr(MN)=\Tr(NM)$ whenever the products are defined, and trace is multiplicative over tensor products;
    \item the ordinary vector Cauchy--Schwarz inequality is $|\braket{u}{v}|\leq\|u\|\,\|v\|$.
\end{enumerate}
These conventions are standard; see, for example, Watrous~\cite{Watrous18}.
The normalized Choi--Jamio{\l}kowski identity below is standard~\cite{Jamiolkowski72,Choi75}; we record the normalization explicitly because conventions differ by factors of the input dimension.

\paragraph{Choi--Jamio\l{}kowski identity.} The following identity expresses a channel's action in terms of its normalized Choi state:

\begin{theorem}\label{thm:basic:choi}
Let
$\Phi_{\regE\rightarrow\regB\regC}:\linear(\regE)\rightarrow
\linear(\regB\otimes\regC)$ be a quantum channel, and fix a copy
$\regA\simeq\regE$ with the corresponding basis. For
$M_{\regA}\in\linear(\regA)$ and
$N_{\regB\regC}\in\linear(\regB\otimes\regC)$, use the same matrix
$M$ on either $\regA$ or $\regE$. Define the maximally entangled state
\[
\ket{\Omega}_{\regA\regE}
=\frac1{\sqrt{\dimension(\regA)}}
\sum_{j\in[\dimension(\regA)]}\ket{j}_{\regA}\ket{j}_{\regE}.
\]
Then
\[
\begin{aligned}
&\frac1{\dimension(\regA)}
\Tr\left(N_{\regB\regC}
\Phi_{\regE\rightarrow\regB\regC}(M_{\regE})\right)\\
&\qquad=
\Tr\left[
(M_{\regA}^{T}\otimes N_{\regB\regC})
(I_{\regA}\otimes\Phi_{\regE\rightarrow\regB\regC})
(\ketbra{\Omega}{\Omega}_{\regA\regE})
\right].
\end{aligned}
\]
\end{theorem}

\begin{proof}
Suppose the $(j,k)^{th}$ entry of $M$ is $\alpha_{j,k}$. Thus, the $(k,j)^{th}$ entry of $M^{T}$ is $\alpha_{j,k}$. 
\begin{eqnarray*}
& & \Tr((M^{T}_{\regA} \otimes N_{\regB \regC})(I_{\regA} \otimes \Phi_{\regE \rightarrow \regB \regC})(\ketbra{\Omega}{\Omega}_{\regA \regE})) \\
& = & \frac{1}{\dimension(\regA)} \sum_{j,k \in [\dimension(\regA)]}  \Tr((M^{T}_{\regA} \otimes N_{\regB \regC})(I_{\regA} \otimes \Phi_{\regE \rightarrow \regB \regC})(\ketbra{j}{k}_{\regA} \otimes \ketbra{j}{k}_{\regE})) \\
& = & \frac{1}{\dimension(\regA)} \sum_{j,k \in [\dimension(\regA)]} \Tr(M^{T}_{\regA} \ketbra{j}{k}_{\regA}) \cdot \Tr(N_{\regB \regC} \Phi_{\regE \rightarrow \regB \regC} (\ketbra{j}{k}_{\regE})) \\ 
& = & \frac{1}{\dimension(\regA)} \sum_{j,k \in [\dimension(\regA)]} \alpha_{j,k} \cdot \Tr(N_{\regB \regC} \Phi_{\regE \rightarrow \regB \regC} (\ketbra{j}{k}_{\regE})) \\
& = & \frac{1}{\dimension(\regA)} \cdot \Tr(N_{\regB \regC} \Phi_{\regE \rightarrow \regB \regC} \left( \sum_{j,k \in [\dimension(\regA)]} \alpha_{j,k} \ketbra{j}{k}_{\regE} \right)) \\
& = & \frac{1}{\dimension(\regA)} \cdot \Tr(N_{\regB \regC} \Phi_{\regE \rightarrow \regB \regC} \left( M_{\regE} \right))
\end{eqnarray*}
\end{proof}

\subsection{Paulis and Cliffords}
\label{sec:paulis:cliffords}
We describe the single-qubit Pauli matrices below:
\[
I=
\begin{pmatrix}
1 & 0\\
0 & 1
\end{pmatrix},
\qquad
X=
\begin{pmatrix}
0 & 1\\
1 & 0
\end{pmatrix},
\qquad
Y=
\begin{pmatrix}
0 & -i\\
i & 0
\end{pmatrix},
\qquad
Z=
\begin{pmatrix}
1 & 0\\
0 & -1
\end{pmatrix}.
\]

\noindent The phase gate is
\[
S=
\begin{pmatrix}
1 & 0\\
0 & i
\end{pmatrix},
\]

\noindent and the Hadamard gate is
\[
H=
\frac{1}{\sqrt{2}}
\begin{pmatrix}
1 & 1\\
1 & -1
\end{pmatrix}.
\]

\noindent We state a basic fact below. 

\begin{fact}
\label{fact:paulis}
Let $P$ be a single-qubit Pauli and let $\ket{\psi}$ be a single-qubit state such that $P \ket{\psi} = \theta \ket{\psi}$, where $\theta \in \{+1,-1\}$. Then measuring $\ket{\psi}$ according to the observable $P$ -- i.e., measuring $\ket{\psi}$ according to the two-outcome projective measurement $\left\{ \underbrace{\frac{I + P}{2}}_{\text{outcome }+1},\ \underbrace{\frac{I - P}{2}}_{\text{outcome }-1} \right\}$ -- yields the outcome $\theta$ with probability 1. 
\end{fact}

\paragraph{Notation.} We use the following notation. For strings $x \in \{0,1\}^{\secparam},z \in \{0,1\}^{\secparam}$,
$$P_{x,z} = \bigotimes_{j \in [\secparam]} P_{x_j,z_j},$$
where $P_{x_j,z_j} = i^{x_jz_j} X^{x_j} Z^{z_j}$. In other words,  
$$P_{x_j,z_j} = \left\{ \begin{array}{cc} I & x_j=z_j=0, \\ X & x_j=1,z_j=0, \\ Z & x_j=0,z_j=1, \\ Y & x_j=1,z_j=1
\end{array} \right.$$
We write $d := 2^{\secparam}$, $P_k := P_{x,z}$ for $k = (x,z)$, and $X^x := \bigotimes_{j} X^{x_j}$, $Z^z := \bigotimes_{j} Z^{z_j}$.

\paragraph{Identities.} The following are well-known identities: 
\begin{enumerate}
\item $HXH=Z$
\item $(SH)^{\dagger} Y (SH)=Z$
\item $X^2 = Y^2 = Z^2 = I$
\item 
$\Tr(X) = \Tr(Y) = \Tr(Z) = 0$
\item 
$\Tr(XY) = \Tr(YX) = \Tr(XZ) = \Tr(ZX) = \Tr(YZ) = \Tr(ZY) = 0$
\end{enumerate}

\paragraph{Hilbert--Schmidt orthogonality.}
\begin{proposition}[Orthogonality of Pauli operators]
\label{lem:security:paulifram}
For every $k=(x,z)\in \F_2^{\secparam} \times \F_2^{\secparam}$, $P_{x,z}$ is Hermitian and $P_{x,z}^2=I$. If $k\neq(0^\secparam,0^\secparam)$, then $P_{x,z}$ also has trace zero. Moreover, for $k=(x,z)$ and $k'=(x',z')$,
\begin{eqnarray}
\Tr(P_{x,z}P_{x',z'})
=
d\cdot\mathbbm{1}[k=k'].
\label{eqn:security:pauliorthogonality}
\end{eqnarray}
The same statements hold for the transposed Paulis $\{P_{x, z}^T: (x, z) \in \F_2^n \times \F_2^n\}$.
\end{proposition}
\begin{proof}
For one qubit, the four matrices $i^{xz}X^xZ^z$ are
$I,Z,X,Y$ for $(x,z)=(0,0),(0,1),(1,0),(1,1)$, respectively.
Each is Hermitian and squares to the identity. Tensor products preserve these
properties, so $P_{x,z}=P_{x,z}^{\dagger}$ and $P_{x,z}^2=I$.

If $k\neq(0^\secparam,0^\secparam)$, at least one local factor is $X$, $Y$, or $Z$, and therefore has trace zero. Trace is multiplicative over tensor products, so $\Tr(P_{x,z})=0$ for every non-identity Pauli.
The one-qubit Paulis are Hilbert--Schmidt orthogonal:
$\Tr(PP')=2\mathbbm{1}[P=P']$. Hence
\begin{eqnarray*}
\Tr(P_{x,z}P_{x',z'})
&=&
\prod_{j=1}^{\secparam}
\Tr(P_{x_j,z_j}P_{x'_j,z'_j})
\\
&=&
2^{\secparam}
\mathbbm{1}[(x,z)=(x',z')]
=
d\mathbbm{1}[k=k'].
\end{eqnarray*}
Every tensor Pauli satisfies $P_{x,z}^{T}=\pm P_{x,z}$, so $P_{x,z}^{T}$ is
also Hermitian, and traceless when $k\neq(0^\secparam,0^\secparam)$. Finally,
\[
\Tr(P_{x,z}^{T}P_{x',z'}^{T})
=
\Tr((P_{x',z'}P_{x,z})^{T})
=
\Tr(P_{x',z'}P_{x,z}),
\]
which proves the orthogonality statement for the transposed Paulis.
\end{proof}

\paragraph{Commutation-anticommutation structure.}

\begin{definition}[Pauli symplectic form]
    For $k = (x, z), k' = (x', z') \in \F_2^n \times \F_2^n$, we define the symplectic form to be:
    \[
      [k,k']=\langle x, z'\rangle+\langle z, x'\rangle\in\F_2.
    \]
\end{definition}

\begin{lemma}[Pauli commutation relations]\label{lem:paulibasic}
    For any $k, k' \in \F_2^n \times \F_2^n$, we have:
    \begin{equation}\label{eq:pauli-basic}
      P_kP_{k'}=(-1)^{[k,k']}P_{k'}P_k.
    \end{equation}
    The same relation holds for the transposed Paulis.
\end{lemma}
\begin{proof}
    Write $k = (x, z)$ and $k' = (x', z')$.
    Additionally, let $|x \land z| \in \N$ denote the number of entries where $x$ and $z$ are both 1.
    Then we have:
    $$P_kP_{k'} = i^{|x \land z| + |x' \land z'|} X^x Z^z X^{x'} Z^{z'} = i^{|x \land z| + |x' \land z'|} \cdot (-1)^{\langle x', z \rangle} X^{x+x'} Z^{z+z'}.$$
    By symmetry, we also have:
    $$P_{k'}P_k = i^{|x \land z| + |x' \land z'|} \cdot (-1)^{\langle x, z' \rangle} X^{x+x'} Z^{z+z'}.$$
    Comparing the above two equalities implies~\eqref{eq:pauli-basic}. Transposing~\eqref{eq:pauli-basic} reverses both products and preserves the sign, which gives the same relation for $P_k^T$.
\end{proof}

\begin{definition}[Symplectic character matrix]\label{def:symp}
    Define $\Xi^{\mathrm{Pauli}} \in \{-1, 1\}^{4^n \times 4^n}$ to be the matrix with rows and columns indexed by elements of $\F_2^n \times \F_2^n$.
    The entry in row $k$ and column $k'$ is $\Xi^{\mathrm{Pauli}}_{k, k'} = (-1)^{[k, k']}$.
\end{definition}

\begin{lemma}[Pauli symplectic structure is spectrally balanced]\label{lem:symplectic-matrix}
    We have $\norm{\Xi^{\mathrm{Pauli}}}_\infty = 2^n$.
\end{lemma}
\begin{proof}
For $k=(x,z)$ and $k'=(x',z')$, the corresponding entry is
\[
\Xi^{\mathrm{Pauli}}_{(x,z),(x',z')}
=(-1)^{\langle x,z'\rangle+\langle z,x'\rangle}.
\]
Permute the columns by exchanging their two coordinates. The resulting
matrix $\Xi'$ has entries
\[
\Xi'_{(x,z),(x',z')}
=(-1)^{\langle x,x'\rangle+\langle z,z'\rangle}.
\]
Thus $2^{-n}\Xi'$ is the normalized $2^{2n}\times2^{2n}$ Hadamard matrix,
which is unitary. Since column permutations preserve operator norm,
\[
\|\Xi^{\mathrm{Pauli}}\|_\infty
=\|\Xi'\|_\infty
=2^n.
\]
\end{proof}

\subsection{Conditions for PSD $2 \times 2$ Block Matrices}

The following is a well-known property of Schur complements; regardless, we include its proof for completeness.

\begin{lemma}[PSD $2 \times 2$ block matrices]\label{lem:2x2}
    Let $A, B, C$ be square Hermitian operators of the same dimension. Then:
    \begin{enumerate}
        \item\label{item:schurblocks} If $C$ is positive-definite, then $\begin{bmatrix} A & B \\ B & C \end{bmatrix}$ is PSD if and only if $A - BC^{-1}B$ is PSD.
        \item\label{item:hadamardblocks} The block matrix $\begin{bmatrix} A & B \\ B & A \end{bmatrix}$ is PSD if and only if $A-B$ and $A+B$ are both PSD.
    \end{enumerate}
\end{lemma}
\begin{proof}
    To prove Item~\ref{item:schurblocks}, 
    note that if we let $T = \begin{bmatrix} I & 0 \\ C^{-1}B & I \end{bmatrix}$, we have:
    \begin{align*}
        T^\dagger \begin{bmatrix} A-BC^{-1}B & 0 \\ 0 & C \end{bmatrix} T &= \begin{bmatrix} I & BC^{-1} \\ 0 & I \end{bmatrix} \begin{bmatrix} A-BC^{-1}B & 0 \\ 0 & C \end{bmatrix}  \begin{bmatrix} I & 0 \\ C^{-1}B & I \end{bmatrix}\\
        &= \begin{bmatrix} I & BC^{-1} \\ 0 & I \end{bmatrix} \begin{bmatrix} A-BC^{-1}B & 0 \\ B & C \end{bmatrix} \\
        &= \begin{bmatrix} A & B \\ B & C \end{bmatrix}.
    \end{align*}
    Since $T$ is invertible, it follows that $\begin{bmatrix} A & B \\ B & C \end{bmatrix}$ being PSD is equivalent to $\begin{bmatrix} A-BC^{-1}B & 0 \\ 0 & C \end{bmatrix}$ being PSD, which is in turn equivalent to $A-BC^{-1}B$ and $C$ both being PSD.
    We are already assuming that $C$ is positive-definite, so this is in turn equivalent to $A-BC^{-1}B$ being PSD.
    
    Now we turn to Item~\ref{item:hadamardblocks}. Note that $U = \frac{1}{\sqrt{2}} \begin{bmatrix} I & I \\ I & -I \end{bmatrix}$ is invertible.
    Conjugating by $U$, we have:
    \begin{align*}
        U \begin{bmatrix} A & B \\ B & A \end{bmatrix} U^\dagger &= \frac{1}{2} \begin{bmatrix} I & I \\ I & -I \end{bmatrix} \begin{bmatrix} A & B \\ B & A \end{bmatrix} \begin{bmatrix} I & I \\ I & -I \end{bmatrix} \\
        &= \frac{1}{2} \begin{bmatrix} I & I \\ I & -I \end{bmatrix} \begin{bmatrix} A+B & A-B \\ A+B & B-A \end{bmatrix} \\
        &= \begin{bmatrix} A+B & 0 \\ 0 & A-B \end{bmatrix}.
    \end{align*}
    So $\begin{bmatrix} A & B \\ B & A \end{bmatrix}$ is PSD if and only if $\begin{bmatrix} A+B & 0 \\ 0 & A-B \end{bmatrix}$ is PSD, which is in turn true if and only if $A+B$ and $A-B$ are both PSD.
\end{proof}

\subsection{Unclonable Encryption and Unclonable Indistinguishability}

We recall the definition of unclonable encryption below; our formulation follows the unclonable IND-CPA definition of Ananth, Kaleoglu, Li, Liu, and Zhandry~\cite{AnanthEtAl22}, building on Broadbent--Lord and Ananth--Kaleoglu~\cite{BroadbentLord20,AnanthKaleoglu21}.

\begin{definition}[One-time private-key unclonable encryption]
A one-time private-key unclonable encryption scheme is a triple of quantum algorithms
$\Pi=(\keygen,\enc,\dec)$ with message space $\mathcal M_{\secparam}$:
\begin{itemize}
    \item $\keygen(1^{\secparam})$ outputs a classical secret key $k$;
    \item $\enc(k,m)$, for $m\in\mathcal M_{\secparam}$, outputs a quantum ciphertext state $\rho_{\ct}$;
    \item $\dec(k,\rho_{\ct})$ outputs a classical message.
\end{itemize}
The scheme is correct if there is a negligible function $\mu$ such that, for every $m\in\mathcal M_{\secparam}$,
\[
\prob\left[\dec(k,\enc(k,m))=m:\ k\xleftarrow{\$}\keygen(1^{\secparam})\right]
\geq 1-\mu(\secparam).
\]
Perfect correctness means that the probability is $1$ for every supported key and every message.
\end{definition}

\begin{definition}[One-time unclonable-indistinguishability game]
\label{def:UI}
Let $\Pi=(\keygen,\enc,\dec)$ be a private-key quantum encryption scheme. An adversary is a tuple
$(\mathcal G,\alice,\bob,\charlie)$, where $\mathcal G$ chooses two equal-length challenge messages,
$\alice$ is the pre-key splitting algorithm, and $\bob,\charlie$ are the non-communicating post-split algorithms. The experiment proceeds as follows:
\begin{enumerate}
    \item $(m_0,m_1,\rho_{\regS})\leftarrow\mathcal G(1^{\secparam})$, with $m_0,m_1\in\mathcal M_{\secparam}$;
    \item Sample $k \leftarrow \keygen(1^{\secparam})$ and $b\xleftarrow{\$}\{0,1\}$;
    \item Prepare $\rho_{\ct}\leftarrow\enc(k,m_b)$ and give $\rho_{\ct}\otimes\rho_{\regS}$ to $\alice$;
    \item $\alice$ applies a channel and outputs a bipartite state $\rho_{\regB\regC}$;
    \item Reveal the same key $k$ to $\bob$ and $\charlie$; they output bits $b_{\bob}$ and $b_{\charlie}$ without communicating;
    \item $(\mathcal G,\alice,\bob,\charlie)$ wins exactly when $b_{\bob}=b_{\charlie}=b$.
\end{enumerate}
We denote the winning probability by
$\prob\left[\operatorname{Win}^{\mathsf{UI}}_{\Pi,\mathcal G,\alice,\bob,\charlie}(\secparam)\right]$.
For the one-bit message space $\{0,1\}$, we fix $(m_0,m_1)=(0,1)$, absorb any auxiliary-state preparation into $\alice$, and write
$\prob\left[\operatorname{Win}^{\mathsf{UI}}_{\Pi,\alice,\bob,\charlie}(\secparam)\right]$
for the winning probability in this specialized game.
\end{definition}

\begin{definition}[Unclonable indistinguishability]
The scheme $\Pi$ is information-theoretically $\varepsilon(\secparam)$-unclonable-indistinguishable if, for every finite-dimensional adversary $(\mathcal G,\alice,\bob,\charlie)$,
\[
\prob\left[\operatorname{Win}^{\mathsf{UI}}_{\Pi,\mathcal G,\alice,\bob,\charlie}(\secparam)\right] 
\leq
\frac{1}{2}+\varepsilon(\secparam).
\]
It is computationally secure if the same inequality holds for every $\qpt$ adversary with
$\varepsilon=\negl$. The baseline $1/2$ is achieved by sending one shared random guess to both recipients. We denote the advantage of $(\mathcal G,\alice,\bob,\charlie)$ to be $\prob\left[\operatorname{Win}^{\mathsf{UI}}_{\Pi,\mathcal G,\alice,\bob,\charlie}(\secparam)\right] - \frac{1}{2}$. 
\end{definition}

\begin{remark}
\label{rem:def:simplification}
For the message space $\{0,1\}$, the fixed-message game above suffices to establish security: equal challenge messages give winning probability at most $1/2$, while the pair $(1,0)$ reduces to $(0,1)$ by complementing both receivers' output bits. Conditioning on the chosen message pair also covers randomized message selection and its associated auxiliary state. Hence, throughout the construction and security proof, we consider adversaries of the form $(\alice,\bob,\charlie)$ in this specialized game.
\end{remark}

\section{Construction}

\noindent We describe a construction of an unclonable encryption scheme for 1-bit messages as follows.
We first describe the ciphertext as a mixed state, then give explicit implementations in~\Cref{sec:impl}.

\begin{itemize}
    \item \underline{$\keygen(1^{\secparam})$}: On input a statistical security parameter $\secparam \geq 1$, it samples and outputs
    $$k = (x, z) \xleftarrow{\$} \left(\{0, 1\}^{\secparam} \times \{0, 1\}^{\secparam}\right) \backslash \{(0^\secparam, 0^\secparam)\}.$$
    \item \underline{$\enc(k,m \in \{0,1\})$}: On input a key $k$ parsed as $(x,z)$ and message $m \in \{0,1\}$, prepare the mixed state:
    $$\rho_\ct^{(k, m)} := \frac{1}{2^\secparam} (I + (-1)^m P_{x, z}).$$
    Note that the above is a well-defined state since $(x, z) \neq (0^\secparam, 0^\secparam)$.

    \item \underline{$\dec(k,\sigma)$}: On input the secret key $k$ parsed as $(x,z)$ and an $\secparam$-qubit ciphertext state $\sigma$, do the following: measure $\sigma$ using the $P_{x, z}$ observable.
    That is, measure $\sigma$ using the two-outcome projective measurement ${\cal M} = \left\{ \underbrace{\frac{I + P_{x, z}}{2}}_{\text{outcome }+1},\ \underbrace{\frac{I - P_{x, z}}{2}}_{\text{outcome }-1} \right\}$ to obtain the outcome $\theta$.
    If $\theta$ is $+1$ then output $0$, else output $1$. 
\end{itemize}

\begin{lemma}
The above scheme satisfies correctness. 
\end{lemma}
\begin{proof}
For $m' \in \{0, 1\}$, we have:
\begin{align*}
    \Pr\left[\dec(k, \rho_\ct^{(k, m)}) = m'\right] &= \Tr\left[\left(\frac{I+(-1)^{m'}P_{x, z}}{2}\right) \rho_\ct^{(k, m)}\right] \\
    &= \Tr\left[\left(\frac{I+(-1)^{m'}P_{x, z}}{2}\right) \left(\frac{I+(-1)^mP_{x, z}}{2^\secparam}\right)\right] \\
    &= \frac{1}{2^{\secparam+1}} \left(\underbrace{\Tr\left[I\right]}_{= 2^{\secparam}} + \left((-1)^{m} + (-1)^{m'}\right) \underbrace{\Tr\left[P_{x, z}\right]}_{= 0} + (-1)^{m+m'}\underbrace{\Tr\left[ I\right]}_{= 2^\secparam}\right) \\
    &= \frac{1}{2}(1+(-1)^{m+m'}),
\end{align*}
which is $1$ if $m' = m$ and 0 otherwise.
\end{proof}

\subsection{Explicit Instantiation}\label{sec:impl}

\paragraph{Encryption.}
To prepare the mixed state $\rho_\ct^{(k, m)}$, do the following: 
\begin{itemize}
    \item For every $j \in [\secparam]$, set $U_{x_j,z_j}$ as:
    $$U_{x_j,z_j} = \left\{ \begin{array}{cc} I & x_j=0 \\
    H & x_j=1,z_j=0 \\ 
    SH & x_j=1,z_j=1
    \end{array}  \right. $$
    Set $U_k = U_{x_1,z_1} \otimes \cdots \otimes U_{x_{\secparam},z_{\secparam}}$. 
    \item By the definition of $\keygen$, there exists a minimal $j^* \in [\secparam]$ for which $P_{x_{j^*}, z_{j^*}} \neq I$ (equivalently, $(x_{j^*}, z_{j^*}) \neq (0, 0)$).
    For each $j \neq j^*$, sample $r_j \xleftarrow{\$} \{0,1\}$.
    Then set $r_{j^*}$ as follows: 
    $$r_{j^*} = m\ \bigoplus\ \left( \bigoplus_{\substack{j \in [\secparam] \\ j \neq j^*}} q_j \cdot r_j \right),$$
    where $q_j$ is defined as follows: 
    $$q_j = \left\{ \begin{array}{cc} 0 & P_{x_j,z_j} = I,\\ 1 & P_{x_j,z_j} \neq I.   \end{array} \right. $$
    (Equivalently, we are sampling $(r_1, \ldots, r_\secparam)$ uniformly subject to the constraint that $\bigoplus_{j \in [n]} q_j \cdot r_j = m$.)
    \item Output the $\secparam$-qubit pure state $\ket{\psi_{k,m;r}} = U_k \ket{r_1 \cdots r_{\secparam}}$.
\end{itemize}
Let us check that this indeed prepares the mixed state $\rho_\ct^{(k, m)}$:
\begin{proposition}
    We have:
    \begin{equation} \label{eqn:main:ct:densmx}
         \rho_{\ct}^{(k,m)} = \frac{1}{2^\secparam} (I + (-1)^m P_{x, z}) = \frac{1}{2^{\secparam-1}} \sum_{\substack{ r \in \{0,1\}^{\secparam}\\ \text{s.t. }m = \bigoplus_{j=1}^{\secparam} q_j r_j }} U_{k} \ketbra{r_1 \cdots r_n}{r_1 \cdots r_n} U_k^{\dagger}.
    \end{equation}
\end{proposition}

\begin{proof}
The vectors $\ket{r_1\cdots r_n}$ with $\bigoplus_jq_jr_j=m$ form an
orthonormal basis of the $(-1)^m$ eigenspace of
$Z^{q_1}\otimes\cdots\otimes Z^{q_n}$. Consequently,
\begin{equation}\label{eqn:ct:density}
\sum_{\substack{r\in\{0,1\}^{\secparam}\\ \bigoplus_jq_jr_j=m}}
\ketbra{r_1\cdots r_n}{r_1\cdots r_n}
=\frac12\left(I+(-1)^m\bigotimes_{j=1}^n Z^{q_j}\right).
\end{equation}
At least one $q_j$ is $1$, so this eigenspace has dimension $2^{n-1}$.
The preparation procedure samples its basis vectors uniformly.
For each qubit,
$U_{x_j,z_j}Z^{q_j}U_{x_j,z_j}^{\dagger}=P_{x_j,z_j}$, and therefore
\[
U_k\left(\bigotimes_{j=1}^n Z^{q_j}\right)U_k^\dagger=P_{x,z}.
\]
Conjugating~\eqref{eqn:ct:density} by $U_k$ and dividing by $2^{n-1}$ gives
\[
\frac1{2^{n-1}}
\sum_{\substack{r\in\{0,1\}^{\secparam}\\ \bigoplus_jq_jr_j=m}}
U_k\ketbra{r_1\cdots r_n}{r_1\cdots r_n}U_k^\dagger
=\frac1{2^n}(I+(-1)^mP_{x,z})
=\rho_{\ct}^{(k,m)}.
\]
This is the claimed ciphertext state.
\end{proof}

\paragraph{Decryption.}
To measure the observable $P_{x, z}$, we use the fact that this observable is in tensor product across the $\secparam$ qubits and accordingly do the following:
\begin{itemize}
    \item For every $j \in [\secparam]$, measure the $j^{th}$ qubit using the $P_{x_j,z_j}$ observable. That is, measure the $j^{th}$ qubit using the two-outcome projective measurement ${\cal M}_j = \left\{ \underbrace{\frac{I + P_{x_j,z_j}}{2}}_{\text{outcome }+1},\ \underbrace{\frac{I - P_{x_j,z_j}}{2}}_{\text{outcome }-1} \right\}$ to obtain the outcome $\theta_j$.
    \item Compute $\theta = \prod_{j=1}^n \theta_j \in \{+1, -1\}$. If $\theta = +1$ then output $0$, else output $1$. (When $P_{x_j,z_j} = I$, the outcome is $\theta_j = +1$ with certainty, so qubit $j$ need not be measured.)
\end{itemize}

\section{Security}
\begin{theorem}\label{thm:onebitsecurity}
Let $\Pi=(\keygen,\enc,\dec)$, where $\keygen$, $\enc$ and $\dec$ are as defined above. For every $n \geq 1$ and every adversary $(\alice,\bob,\charlie)$, we have:
\begin{equation}\label{eq:sectarget}
    \prob\left[\operatorname{Win}^{\mathsf{UI}}_{\Pi,\alice,\bob,\charlie}(\secparam)\right] \leq \frac{1}{2} + \frac{1}{2} \sqrt{\frac{2^n}{4^n-1}},
\end{equation}
where we use the one-bit specialization of Definition~\ref{def:UI}, as justified in Remark~\ref{rem:def:simplification}.
\end{theorem}
\noindent
The proof has four parts. First, we express the winning probability in
terms of the largest eigenvalue of an operator. Next, we use the Pauli
structure to bound the product of Bob's and Charlie's correlation
operators. We then show that the bound still holds after inserting the
operator that accounts for disagreement. Finally, we use this estimate
to bound the largest eigenvalue and hence the winning probability.

\subsection{Part 1: Ricochet Reduction}
Fix an adversary $(\alice,\bob,\charlie)$. By
Remark~\ref{rem:def:simplification}, we may take the challenge messages to be
$0$ and $1$. Let
$\Phi_{\regE\rightarrow\regB\regC}:\linear(\regE)\rightarrow
\linear(\regB\otimes\regC)$ be the cloner's channel, where $\regE$
is the ciphertext register and $\regB,\regC$ are sent to Bob and Charlie.
Any auxiliary state prepared before receiving the ciphertext can be
included in the definition of $\Phi$.
For each key $k$, write the receivers' binary POVMs as
$\{\Lambda_{\bob,0}^{(k)},\Lambda_{\bob,1}^{(k)}\}$ and
$\{\Lambda_{\charlie,0}^{(k)},\Lambda_{\charlie,1}^{(k)}\}$, and set
\[
B_k=\Lambda_{\bob,0}^{(k)}-\Lambda_{\bob,1}^{(k)},
\qquad
C_k=\Lambda_{\charlie,0}^{(k)}-\Lambda_{\charlie,1}^{(k)}.
\]
Thus their measurement operators for a guess $m$ are
$(I_{\regB}+(-1)^mB_k)/2$ and $(I_{\regC}+(-1)^mC_k)/2$.
The observables $B_k,C_k$ are Hermitian contractions, and the winning
probability is
\begin{equation}\label{eqn:winningprob}
\begin{split}
p&=\prob[(\alice,\bob,\charlie)\text{ wins}]\\
&=\expect_{k=(x,z),m}
\Tr\left[
\left(\frac{I_{\regB}+(-1)^mB_k}{2}
\otimes\frac{I_{\regC}+(-1)^mC_k}{2}\right)
\Phi_{\regE\rightarrow\regB\regC}(\rho_{\ct}^{(k,m)})
\right].
\end{split}
\end{equation}
The averages are over the sampled key and a uniform bit $m$.

Let $\regA$ be a copy of $\regE$, and let
\[
\tau_{\regA\regB\regC}
=(I_{\regA}\otimes\Phi_{\regE\rightarrow\regB\regC})
(\ketbra{\Omega}{\Omega}_{\regA\regE})
\]
be the normalized Choi state. Applying \Cref{thm:basic:choi} and the
ciphertext formula~\eqref{eqn:main:ct:densmx} to
\eqref{eqn:winningprob} gives
\begin{equation}\label{eqn:densmx:tensorproduct}
p=\expect_{k=(x,z),m}
\Tr\left[
\left((I+(-1)^mP_{x,z}^T)_{\regA}
\otimes\frac{I_{\regB}+(-1)^mB_k}{2}
\otimes\frac{I_{\regC}+(-1)^mC_k}{2}\right)
\tau_{\regA\regB\regC}
\right].
\end{equation}
The terms with an odd power of $(-1)^m$ vanish on averaging over $m$.
The four remaining terms give
\begin{equation}\label{eqn:winprob:messageaveraged}
\begin{split}
p=\frac14\expect_{k=(x,z)}\Tr\Big[
\big(&I_{\regA\regB\regC}
+(P_{x,z}^T)_{\regA}\otimes B_k\otimes I_{\regC}\\
&+(P_{x,z}^T)_{\regA}\otimes I_{\regB}\otimes C_k
+I_{\regA}\otimes B_k\otimes C_k\big)
\tau_{\regA\regB\regC}\Big].
\end{split}
\end{equation}
Define
\begin{equation}\label{eqn:definingG}
\begin{split}
\bfG=\frac12\expect_{k=(x,z)}\Big[
&(P_{x,z}^T)_{\regA}\otimes B_k\otimes I_{\regC}
+(P_{x,z}^T)_{\regA}\otimes I_{\regB}\otimes C_k\\
&+I_{\regA}\otimes B_k\otimes C_k-I_{\regA\regB\regC}
\Big].
\end{split}
\end{equation}
Since $\tau$ has trace one, we obtain
\begin{equation}\label{eq:endofpart1}
p=\frac12\left(1+\Tr(\bfG\tau)\right)
\leq\frac12+\frac12\lambda_{\max}(\bfG),
\end{equation}
where the inequality is~\eqref{eq:maxeig}. It remains to bound the largest
eigenvalue of $\bfG$.

\begin{definition}[Keys and decoding observables]\label{def:decoders}
    For the rest of the proof, let
\[
\begin{gathered}
\cK_{\secparam}
:=
\left(\{0,1\}^{\secparam}\times\{0,1\}^{\secparam}\right)
\setminus\{(0^\secparam,0^\secparam)\},
\\
d:=2^{\secparam},
\qquad
K:=|\cK_{\secparam}|
=
4^\secparam - 1
=
d^2 - 1.
\end{gathered}
\]
For $k=(x,z)\in\cK_{\secparam}$, write $\Theta_k$ for the transposed Pauli on the reference register. Recall the decoding observables from the reduction:
\begin{eqnarray}
\Theta_k:=(P_{x,z}^{T})_{\regA},
\qquad
B_k\in\linear(\regB),
\qquad
C_k\in\linear(\regC).
\label{eqn:security:shorthand}
\end{eqnarray}
Since $B_k,C_k$ arise from binary POVMs, they are self-adjoint contractions:
\begin{eqnarray}
B_k=B_k^{\dagger},
\quad
C_k=C_k^{\dagger},
\quad
-I_{\regB}\preceq B_k\preceq I_{\regB},
\quad
-I_{\regC}\preceq C_k\preceq I_{\regC}.
\label{eqn:security:decodercontractions}
\end{eqnarray}
\end{definition}

\begin{definition}[Averaged operators]\label{def:averaged}
Define the following operators on $\regA\otimes\regB\otimes\regC$:
\begin{eqnarray}
\mathsf{E}_{B}
&:=&
\frac{1}{K}
\sum_{k\in\cK_{\secparam}}
\Theta_k\otimes B_k\otimes I_{\regC},
\label{eqn:security:EB}
\\
\mathsf{E}_{C}
&:=&
\frac{1}{K}
\sum_{k\in\cK_{\secparam}}
\Theta_k\otimes I_{\regB}\otimes C_k,
\label{eqn:security:EC}
\\
\mathsf{E}_{BC}
&:=&
\frac{1}{K}
\sum_{k\in\cK_{\secparam}}
I_{\regA}\otimes B_k\otimes C_k.
\label{eqn:security:EBC}
\\
\sfD
&:=&
I_{\regA\regB\regC} - \sfE_{BC}.
\label{eqn:security:D}
\end{eqnarray}
The operator $\bfG$ defined in~\Cref{eqn:definingG} can equivalently be written
as
\begin{eqnarray}
\bfG
=
\frac{1}{2}
\left(
\mathsf{E}_{B}
+
\mathsf{E}_{C}
+
\mathsf{E}_{BC}
-
I_{\regA\regB\regC}
\right).
\label{eqn:security:Gcorrelations}
\end{eqnarray}
Each of $\sfE_B,\sfE_C,\sfE_{BC}$ is an average of Hermitian contractions,
so it is itself a Hermitian contraction. In particular,
$0\preceq\sfD\preceq2I$. For the Choi state $\tau$ of the fixed attack,
\begin{align*}
\Tr(\sfE_B\tau)&=2\prob[\text{Bob guesses correctly}]-1,\\
\Tr(\sfE_C\tau)&=2\prob[\text{Charlie guesses correctly}]-1,\\
\Tr(\sfE_{BC}\tau)&=\prob[\text{Bob and Charlie agree}]
-\prob[\text{Bob and Charlie disagree}],\\
\Tr(\sfD\tau)&=2\prob[\text{Bob and Charlie disagree}].
\end{align*}
These identities follow by assigning signs $+1,-1$ to the binary
measurement outcomes, as in the reduction above.
\end{definition}

\subsection{Part 2: Bounding the Product of the Correlation Operators}

We now use the Pauli structure to show that the product of the two correlation operators has small norm.
\begin{lemma}[Product of the correlation operators]
\label{prop:security:endpoint}
Let $\gamma := d/K$. Then
\begin{eqnarray}
\left\|
\mathsf{E}_{B}\mathsf{E}_{C}
\right\|_{\infty}
\leq \gamma.
\label{eqn:security:endpoint}
\end{eqnarray}
\end{lemma}
\noindent
We present two proofs; the first is adapted from~\cite{ragavan2026efficientunclonableencryptionpauli} and the other is from~\cite{ananth2026unconditionalunclonableencryption}.

\subsubsection{Proof A: Using Pauli Commutation-Anticommutation Structure}

\begin{proof}[Proof A of Lemma~\ref{prop:security:endpoint}]
We will make use of the matrix $\Xi^{\mathrm{Pauli}}$ defined in Definition~\ref{def:symp}.
Recalling the definitions from Definition~\ref{def:averaged}, we compute:
\begin{align*}
    \sfE_B \sfE_C &= \frac{1}{K^2} \sum_{k, k' \in \cK_\secparam} \left(\Theta_k\otimes B_k\otimes I_{\regC}\right) \left(\Theta_{k'}\otimes I_{\regB}\otimes C_{k'}\right) \\
    &= \frac{1}{K^2} \sum_{k, k' \in \cK_{\secparam}} \Theta_k \Theta_{k'} \otimes B_k \otimes C_{k'} \\
    &= \frac{1}{K^2} \sum_{k, k' \in \cK_{\secparam}} \Xi^{\mathrm{Pauli}}_{k', k} \Theta_{k'} \Theta_{k} \otimes B_k \otimes C_{k'} \\
    &= \frac{1}{K^2} \sum_{k, k' \in \cK_{\secparam}} \Xi^{\mathrm{Pauli}}_{k', k} \left(\Theta_{k'} \otimes I_{\regB} \otimes C_{k'}\right) \left( \Theta_{k} \otimes B_k \otimes I_{\regC}\right).
\end{align*}
To help rewrite the above expression, let $\tilde{d}$ be the dimension of the Hilbert space of $\regA\regB \regC$, and let $M_C \in \C^{\tilde{d} \times \tilde{d}|\cK_\secparam|}$ be a block matrix with $|\cK_\secparam|$ block columns indexed by $k' \in \cK_\secparam$. The block column corresponding to $k'$ will be $\Theta_{k'} \otimes I_{\regB} \otimes C_{k'}$.
Similarly define $M_B \in \C^{\tilde{d}|\cK_\secparam| \times \tilde{d}}$ to have row blocks indexed by $k \in \cK_\secparam$ and containing $\Theta_{k} \otimes B_k \otimes I_{\regC}$ in the block corresponding to $k$.
Finally, let $M_\Xi \in \C^{\tilde{d}|\cK_n| \times \tilde{d}|\cK_n|}$ be defined to be:
$$M_\Xi := \Xi^{\mathrm{Pauli}}_{\cK_\secparam, \cK_\secparam} \otimes I_{\tilde{d}}.$$
In words, $M_\Xi$ has block rows and columns indexed by $k', k \in \cK_\secparam$ respectively. The corresponding $\tilde{d} \times \tilde{d}$ block is $\Xi^{\mathrm{Pauli}}_{k', k}$ times the identity matrix.
The restriction $\Xi^{\mathrm{Pauli}}_{\cK_\secparam,\cK_\secparam}$ deletes the row and column indexed by the identity Pauli.
In these terms, the outcome of the above computation can be rewritten as:
\begin{align*}
    \sfE_B \sfE_C &= \frac{1}{K^2} M_C M_\Xi M_B \\
    &= \frac{1}{K^2} M_C (\Xi^{\mathrm{Pauli}}_{\cK_\secparam, \cK_\secparam} \otimes I_{\tilde{d}}) M_B \\
    \Rightarrow \norm{\sfE_B \sfE_C}_\infty &= \frac{1}{K^2} \norm{M_C (\Xi^{\mathrm{Pauli}}_{\cK_\secparam, \cK_\secparam} \otimes I_{\tilde{d}}) M_B}_\infty \\
    &\leq \frac{1}{K^2} \norm{M_C}_\infty \cdot \norm{\Xi^{\mathrm{Pauli}}_{\cK_\secparam, \cK_\secparam}}_\infty \cdot \norm{M_B}_\infty.
\end{align*}
We bound each term of the RHS above as follows:
\begin{itemize}
    \item For any unit $\ket{\psi}$, the vector $M_B \ket{\psi}$ will be the concatenation of $\left(\Theta_{k} \otimes B_k \otimes I_{\regC}\right) \ket{\psi}$ for every $k \in \cK_\secparam$.
    Since $\Theta_k$ is unitary and $B_k$ a contraction, each of these chunks has norm at most 1, so the total $\ell_2$ norm will be at most $\sqrt{|\cK_\secparam|} = \sqrt{K}$.
    Consequently, we have $\norm{M_B}_\infty \leq \sqrt{K}$.
    \item Applying the same argument to $M_C^\dagger$, whose row blocks are $\Theta_{k'} \otimes I_{\regB} \otimes C_{k'}$, gives $\norm{M_C}_\infty = \norm{M_C^\dagger}_\infty \leq \sqrt{K}$.
    \item Finally, we have $\norm{\Xi^{\mathrm{Pauli}}_{\cK_\secparam, \cK_\secparam}}_\infty \leq \norm{\Xi^{\mathrm{Pauli}}}_\infty$ since taking a submatrix cannot increase operator norm. The latter quantity is exactly $2^n = d$ by Lemma~\ref{lem:symplectic-matrix}.
\end{itemize}
Putting these together yields:
$$\norm{\sfE_B \sfE_C}_\infty \leq \frac{1}{K^2} \cdot K \cdot d = \frac{d}{K},$$
as desired.
\end{proof}

\subsubsection{Proof B: Using Conditional Overlap and Hilbert--Schmidt Orthogonality}
The second proof uses the following inequality for operators that share
one register. It bounds their product using partial traces over that
register.

\begin{lemma}[Conditional overlap]\label{lem:conditionaloverlap}
Let $U_{\regA\regB}\in\linear(\regA\otimes\regB)$ and
$V_{\regA\regC}\in\linear(\regA\otimes\regC)$. Extend each operator to
$\regA\otimes\regB\otimes\regC$ by the identity on its missing register,
and denote these extensions by $\widetilde U_{\regA\regB\regC}$ and
$\widetilde V_{\regA\regB\regC}$. Then
\[
\|\widetilde U_{\regA\regB\regC}
\widetilde V_{\regA\regB\regC}\|_\infty
\leq
\sqrt{\|\Tr_{\regA}(U_{\regA\regB}^\dagger U_{\regA\regB})\|_\infty}
\sqrt{\|\Tr_{\regA}(V_{\regA\regC}V_{\regA\regC}^\dagger)\|_\infty}.
\]
\end{lemma}
\begin{proof}
Write $d_{\regA}=\dim(\regA)$ and decompose the operators into blocks:
\[
U_{\regA\regB}
=\sum_{i,j\in[d_{\regA}]}\ketbra{i}{j}_{\regA}\otimes(U_{i,j})_{\regB},
\qquad
V_{\regA\regC}
=\sum_{j,k\in[d_{\regA}]}\ketbra{j}{k}_{\regA}\otimes(V_{j,k})_{\regC}.
\]
The extensions are therefore
\begin{align*}
\widetilde U_{\regA\regB\regC}
&=\sum_{i,j}\ketbra{i}{j}_{\regA}\otimes(U_{i,j})_{\regB}\otimes I_{\regC},\\
\widetilde V_{\regA\regB\regC}
&=\sum_{j,k}\ketbra{j}{k}_{\regA}\otimes I_{\regB}\otimes(V_{j,k})_{\regC}.
\end{align*}
For each $j,k$, set
\[
T_{j,k}^{\regA\regB}
=\sum_i\ketbra{i}{k}_{\regA}\otimes(U_{i,j})_{\regB}.
\]
Multiplying the blocks gives
\[
\widetilde U_{\regA\regB\regC}\widetilde V_{\regA\regB\regC}
=\sum_{j,k}T_{j,k}^{\regA\regB}\otimes(V_{j,k})_{\regC}.
\]
Fix unit vectors $\ket{\xi},\ket{\eta}$ on
$\regA\otimes\regB\otimes\regC$.
Since $T_{j,k}^{\regA\regB}$ and $(V_{j,k})_{\regC}$ act on disjoint
registers, Cauchy--Schwarz gives
\begin{align*}
&\left|\bra{\eta}\widetilde U_{\regA\regB\regC}
\widetilde V_{\regA\regB\regC}\ket{\xi}\right|^2\\
&\quad=
\left|\sum_{j,k}
\bra{\eta}(I_{\regA\regB}\otimes(V_{j,k})_{\regC})
(T_{j,k}^{\regA\regB}\otimes I_{\regC})\ket{\xi}\right|^2\\
&\quad\leq
\left(\sum_{j,k}
\|(I_{\regA\regB}\otimes(V_{j,k})_{\regC}^{\dagger})\ket{\eta}\|^2\right)
\left(\sum_{j,k}
\|(T_{j,k}^{\regA\regB}\otimes I_{\regC})\ket{\xi}\|^2\right).
\end{align*}
The first sum is bounded by the second partial trace in the statement:
\begin{align*}
\sum_{j,k}
\|(I_{\regA\regB}\otimes(V_{j,k})_{\regC}^{\dagger})\ket{\eta}\|^2
&=\bra{\eta}
\left(I_{\regA\regB}\otimes
\sum_{j,k}(V_{j,k})_{\regC}(V_{j,k})_{\regC}^{\dagger}\right)
\ket{\eta}\\
&=\bra{\eta}\left(I_{\regA\regB}\otimes
\Tr_{\regA}(V_{\regA\regC}V_{\regA\regC}^{\dagger})\right)\ket{\eta}\\
&\leq\|\Tr_{\regA}(V_{\regA\regC}V_{\regA\regC}^{\dagger})\|_\infty.
\end{align*}
For the other sum, direct multiplication gives
\begin{align*}
\sum_{j,k}(T_{j,k}^{\regA\regB})^\dagger T_{j,k}^{\regA\regB}
&=\sum_{i,j,k}\ketbra{k}{k}_{\regA}
\otimes(U_{i,j})_{\regB}^\dagger(U_{i,j})_{\regB}\\
&=I_{\regA}\otimes\Tr_{\regA}(U_{\regA\regB}^\dagger U_{\regA\regB}).
\end{align*}
Hence
\[
\sum_{j,k}\|(T_{j,k}^{\regA\regB}\otimes I_{\regC})\ket{\xi}\|^2
\leq\|\Tr_{\regA}(U_{\regA\regB}^\dagger U_{\regA\regB})\|_\infty.
\]
Combining the two bounds and taking the supremum over unit
$\ket{\xi},\ket{\eta}$ proves the lemma.
\end{proof}

\begin{proof}[Proof B of Lemma~\ref{prop:security:endpoint}]
To apply the lemma, write the two correlation operators on their respective registers:
\[
U_{\regA\regB}
:=
\frac{1}{K}
\sum_{k\in\cK_{\secparam}}
\Theta_k\otimes B_k,
\qquad
V_{\regA\regC}
:=
\frac{1}{K}
\sum_{k\in\cK_{\secparam}}
\Theta_k\otimes C_k.
\]
Their canonical embeddings in
$\regA\otimes\regB\otimes\regC$ are exactly
$\mathsf{E}_{B}$ and $\mathsf{E}_{C}$.

Using~\Cref{eqn:security:pauliorthogonality} and preserving the order of the
Bob operators,
\begin{eqnarray*}
\Tr_{\regA}
\left(
U_{\regA\regB}^{\dagger}U_{\regA\regB}
\right)
&=&
\frac{1}{K^2}
\sum_{k,k'}
\Tr(\Theta_k\Theta_{k'})
B_kB_{k'}
\\
&=&
\frac{d}{K^2}
\sum_k B_k^2
\preceq
\frac{d}{K}I_{\regB}.
\end{eqnarray*}
Here $B_k^2\preceq I_{\regB}$ because $B_k$ is a Hermitian contraction.
\par Similarly,
\[
\Tr_{\regA}
\left(
V_{\regA\regC}V_{\regA\regC}^{\dagger}
\right)
=
\frac{d}{K^2}
\sum_k C_k^2
\preceq
\frac{d}{K}I_{\regC}.
\]
Applying the conditional-overlap lemma (Lemma~\ref{lem:conditionaloverlap}) proved above gives
\[
\left\|
\mathsf{E}_{B}\mathsf{E}_{C}
\right\|_{\infty}
\leq
\sqrt{\frac{d}{K}}
\sqrt{\frac{d}{K}}
=
\frac{d}{K}.
\]
\end{proof}

\subsection{Part 3: Accounting for Disagreement}
The product bound alone does not control joint success: the receivers'
individual advantages may come from different events, as in the
random-forwarding example in the overview. We therefore insert an
operator that gives more weight to directions with less disagreement.
The next proposition lets us do this using a power series in
$\mathsf E_{BC}$.

\begin{proposition}[Inserting powers of the agreement operator]
\label{prop:security:moments}
For every integer $\ell\geq0$,
\begin{equation}\label{eqn:security:moments}
\|\mathsf E_B\mathsf E_{BC}^{\ell}\mathsf E_C\|_\infty
\leq\|\mathsf E_B\mathsf E_C\|_\infty
\leq\gamma.
\end{equation}
\end{proposition}
\begin{proof}
For $\ell\geq1$, expand the inserted power and move only operators on
disjoint registers:
\begin{equation}\label{eqn:security:momentidentity}
\begin{split}
\mathsf E_B\mathsf E_{BC}^{\ell}\mathsf E_C
=\frac1{K^\ell}
\sum_{k_1,\ldots,k_\ell\in\cK_\secparam}
&\left(I_{\regA}\otimes I_{\regB}\otimes
C_{k_1}\cdots C_{k_\ell}\right)
\mathsf E_B\mathsf E_C\\
{}\cdot&
\left(I_{\regA}\otimes B_{k_1}\cdots B_{k_\ell}
\otimes I_{\regC}\right).
\end{split}
\end{equation}
The order of the Bob operators, and separately of the Charlie operators,
is unchanged. Each outside product has norm at most one, since every
$B_k,C_k$ is a contraction. The triangle inequality therefore gives
\[
\|\mathsf E_B\mathsf E_{BC}^{\ell}\mathsf E_C\|_\infty
\leq\frac1{K^\ell}
\sum_{k_1,\ldots,k_\ell\in\cK_\secparam}
\|\mathsf E_B\mathsf E_C\|_\infty
=\|\mathsf E_B\mathsf E_C\|_\infty
\leq\gamma.
\]
The case $\ell=0$ is exactly the product bound of
Lemma~\ref{prop:security:endpoint}.
\end{proof}

\begin{definition}[The filter $\sfF$]
    Let $t > 0$ be a real number to be chosen later.
    Recall that $\sfD$ is positive semidefinite and measures disagreement. Write its spectral decomposition as
\[
\sfD=\sum_j d_j\ketbra{v_j}{v_j},
\qquad d_j\geq0.
\]
Define
\begin{eqnarray}
\sfF(t):=2t(2tI + \sfD)^{-1}=\sum_j\frac{2t}{2t+d_j}\ketbra{v_j}{v_j}.
\label{eqn:security:filter}
\end{eqnarray}
Write $\sfF=\sfF(t)$ when $t$ is fixed. Its eigenvalues $2t/(2t+d_j)$ lie in $(0,1]$ and decrease as $d_j$ increases. Thus $\sfF$ gives more weight to directions with less disagreement, and $0\prec\sfF\preceq I$.
The same diagonal representation gives
\begin{eqnarray}
(2tI+\sfD)\sfF=\sfF(2tI+\sfD)=2tI.
\label{eqn:security:filterproperties}
\end{eqnarray}
In particular, $\sfF$ commutes with $\sfD$ and with $\mathsf E_{BC}=I-\sfD$; no commutation with $\mathsf E_B$ or $\mathsf E_C$ is asserted.
\end{definition}
\noindent
The preceding power bounds imply that inserting $\sfF(t)$ preserves the product estimate:
\begin{lemma}\label{lem:ebfec}
    For any choice of $t > 0$, we have:
    \begin{eqnarray}
\left\|
\mathsf{E}_{B}\sfF(t)\mathsf{E}_{C}
\right\|_{\infty}
\leq
\gamma.
\label{eqn:security:filteredupper}
\end{eqnarray}
\end{lemma}
\noindent

\noindent
\begin{proof}
We will prove this by considering a power series expansion of $\sfF(t)$.
Indeed, we have:
\begin{align*}
    2tI+\sfD &= 2tI + (I - (I-\sfD)) \\
    &= (2t+1)\left(I - \frac{I - \sfD}{2t+1}\right) \\
    &= (2t+1)\left(I - \frac{\sfE_{BC}}{2t+1}\right).
\end{align*}
Since $\sfE_{BC}$ is an average of Hermitian contractions, we have $\norm{\sfE_{BC}}_\infty \leq 1 \Rightarrow \norm{\sfE_{BC}/(2t+1)}_\infty < 1$.
Therefore we can indeed express the inverse with a power series as follows:
\begin{align}
    (2tI+\sfD)^{-1} &= (2t+1)^{-1} \sum_{\ell = 0}^\infty \left(\frac{\sfE_{BC}}{2t+1}\right)^\ell \nonumber \\
    \Rightarrow \sfF &= \frac{2t}{2t+1} \sum_{\ell = 0}^\infty \left(\frac{\sfE_{BC}}{2t+1}\right)^\ell.\label{eq:ftaylorexpansion}
\end{align}
Hence we have:
\begin{align*}
    \norm{\sfE_B \sfF \sfE_C}_\infty &= \norm{\frac{2t}{2t+1} \sum_{\ell = 0}^\infty \frac{1}{(2t+1)^\ell} \sfE_B \sfE_{BC}^\ell \sfE_C}_\infty \\
    &\leq \frac{2t}{2t+1} \sum_{\ell = 0}^\infty \frac{1}{(2t+1)^\ell} \norm{\sfE_B \sfE_{BC}^\ell \sfE_C}_\infty \\
    &\leq \frac{2t}{2t+1} \sum_{\ell = 0}^\infty \frac{\gamma}{(2t+1)^\ell} \text{ (using Proposition~\ref{prop:security:moments})} \\
    &= \frac{2t}{2t+1} \cdot \gamma \cdot \frac{1}{1 - \frac{1}{2t+1}} \\
    &= \gamma,
\end{align*}
as desired.
\end{proof}

\subsection{Part 4: Completing the Proof}

We now relate the largest eigenvalue of $\bfG$ to the quantity just bounded.
\begin{lemma}\label{lem:maxeig}\label{prop:security:positivepart}
    Let $t = \lambda_{\max}(\bfG)$.
    If $t > 0$, then we have $t^2 \leq \norm{\sfE_B \sfF(t) \sfE_C}_\infty$.
\end{lemma}
\noindent
Let us first see why this would suffice to complete the security proof:

\begin{proof}[Proof of Theorem~\ref{thm:onebitsecurity}, assuming Lemma~\ref{lem:maxeig}]
    Let $t=\lambda_{\max}(\bfG)$. We claim that $t \leq \sqrt{d/K}$.
    If $t \leq 0$ this is automatic.
    If $t > 0$, then combining Lemma~\ref{lem:maxeig} with Lemma~\ref{lem:ebfec} implies that:
    \[
        t^2 \leq \norm{\sfE_B \sfF(t) \sfE_C}_\infty \leq \frac{d}{K} \Rightarrow t \leq \sqrt{\frac{d}{K}}.
    \]
Substituting this into~\Cref{eq:endofpart1} gives:
\begin{eqnarray*}
p
&\leq&
\frac{1}{2}
+
\frac{1}{2}
\sqrt{\frac{d}{K}}
\\
&=&
\frac{1}{2}
+
\frac{1}{2} \sqrt{\frac{2^\secparam}{4^\secparam - 1}}.
\end{eqnarray*}
This proves
\Cref{eq:sectarget} and completes the proof.
\end{proof}
\noindent
To prove Lemma~\ref{lem:maxeig}, recall from~\Cref{sec:techoverview} that the high-level idea is to use the following identity:
\begin{equation}\label{eq:matdiffsqs}
    (\mathsf E_B+\mathsf E_C)\sfF(t)
(\mathsf E_B+\mathsf E_C) = (\mathsf E_B-\mathsf E_C)\sfF(t)
(\mathsf E_B-\mathsf E_C) + 2 \mathsf{E}_B \mathsf{F}(t) \mathsf{E}_C + 2\mathsf{E}_C \mathsf{F}(t) \mathsf{E}_B.
\end{equation}
The latter two terms of the RHS are handled by the norm $\norm{\sfE_B \sfF \sfE_C}_\infty$. To handle the first term on the RHS, we have a few lemmas:

\begin{proposition}\label{prop:disagreement}
    Let $\sfD$ be as defined in~\Cref{eqn:security:D}.
Then we have:
\begin{eqnarray}
-\sfD\preceq\sfE_B - \sfE_C\preceq\sfD.
\label{eqn:security:orderinterval}
\end{eqnarray}
\end{proposition}
\begin{proof}
    For every $u,v\in\{+1,-1\}$, define
\begin{eqnarray}
\Gamma^{u,v}
:=
\frac{1}{4}
\left(
I_{\regA\regB\regC}
+
u\mathsf{E}_{B}
+
v\mathsf{E}_{C}
+
uv\mathsf{E}_{BC}
\right).
\label{eqn:security:Gamma}
\end{eqnarray}
We claim that $\Gamma^{u,v} \succeq0$ for all $u, v$.
To see this, for a fixed key $k$, define $W_{k,B} := \Theta_k\otimes B_k\otimes I_{\regC}$, and $W_{k,C} :=
\Theta_k\otimes I_{\regB}\otimes C_k$.
These are Hermitian operators and have operator norm at most 1. They commute, because $W_{k,B}W_{k,C} = \Theta_k^2\otimes B_k\otimes C_k = I_{\regA}\otimes B_k\otimes C_k = W_{k,C}W_{k,B}.$
Thus $I+uW_{k,B}$ and $I+vW_{k,C}$ are commuting positive semidefinite operators (which means that they can be simultaneously diagonalized), so
\[
\Gamma_k^{u,v}
:=
\frac{1}{4}
(I+uW_{k,B})(I+vW_{k,C})
\succeq 0.
\]
Expanding and averaging uniformly over $k$ yields $\Gamma^{u, v} \succeq 0$ as claimed.
Now the conclusion follows from noting that:
\[
4\Gamma^{+,-}=\sfD+(\sfE_B-\sfE_C),
\qquad
4\Gamma^{-,+}=\sfD-(\sfE_B - \sfE_C),
\]
and we have established that $\Gamma^{+, -}$ and $\Gamma^{-, +}$ are both positive semidefinite.
\end{proof}

\begin{proposition}[Quadratic sandwich bound]
\label{clm:security:quadraticsandwich}
For any $t > 0$, we have $(\sfE_B - \sfE_C)\sfF(t)(\sfE_B - \sfE_C)\preceq 2t\sfD$.
\end{proposition}

\begin{proof}
By Proposition~\ref{prop:disagreement} and Item~\ref{item:hadamardblocks} of Lemma~\ref{lem:2x2}, the following matrix:
$$\begin{bmatrix} \sfD & \sfE_B-\sfE_C \\ \sfE_B - \sfE_C & \sfD \end{bmatrix}$$
is PSD.
Consequently, we also have:
$$\begin{bmatrix} \sfD & \sfE_B-\sfE_C \\ \sfE_B-\sfE_C & 2tI + \sfD\end{bmatrix} = \begin{bmatrix} \sfD & \sfE_B - \sfE_C \\ \sfE_B-\sfE_C & \sfD \end{bmatrix} + \begin{bmatrix} 0 & 0 \\ 0 & 2tI \end{bmatrix} \succeq 0.$$
Now applying Item~\ref{item:schurblocks} of Lemma~\ref{lem:2x2} (noting that $2tI+\sfD$ is positive definite) implies that:
\begin{align*}
    \sfD - (\sfE_B-\sfE_C) (2tI+\sfD)^{-1} (\sfE_B-\sfE_C) &\succeq 0 \\
    \Leftrightarrow 2t\sfD &\succeq (\sfE_B - \sfE_C) 2t(2tI+\sfD)^{-1} (\sfE_B - \sfE_C) \\
    &= (\sfE_B - \sfE_C) \sfF (\sfE_B - \sfE_C),
\end{align*}
as desired.
\end{proof}

\begin{corollary}\label{cor:sumupperbound}
    We have $(\sfE_B + \sfE_C)\sfF(t)(\sfE_B + \sfE_C) \preceq 2t\sfD + 2\sfE_B \sfF \sfE_C + 2\sfE_C \sfF \sfE_B$.    
\end{corollary}
\begin{proof}
    This is immediate from~\eqref{eq:matdiffsqs} after plugging in Proposition~\ref{clm:security:quadraticsandwich}.
\end{proof}

\medskip
\noindent
With this, we can complete the proof of Lemma~\ref{lem:maxeig}:
\begin{proof}[Proof of Lemma~\ref{lem:maxeig}]
Assume $t=\lambda_{\max}(\bfG)>0$ and
choose a unit eigenvector $\ket{\psi}$ on
$\regA\otimes\regB\otimes\regC$ such that
\begin{equation}\label{eqn:security:topeigenvector}
\bfG\ket{\psi}=t\ket{\psi}.
\end{equation}
Using~\eqref{eqn:security:Gcorrelations}, this eigenvector equation becomes
\begin{equation}\label{eqn:security:sumeigenvector}
(\mathsf E_B+\mathsf E_C)\ket{\psi}
=(2tI+\sfD)\ket{\psi}.
\end{equation}
Recall that $(2tI+\sfD)\sfF=2tI$. Applying the eigenvector equation on
both sides of the quadratic form gives
\begin{equation}\label{eqn:security:sumquadratic}
\begin{split}
\bra{\psi}(\mathsf E_B+\mathsf E_C)\sfF
(\mathsf E_B+\mathsf E_C)\ket{\psi}
&=\bra{\psi}(2tI+\sfD)\sfF(2tI+\sfD)\ket{\psi}\\
&=4t^2+2t\bra{\psi}\sfD\ket{\psi}.
\end{split}
\end{equation}
On the other hand, Corollary~\ref{cor:sumupperbound} implies that:
\begin{equation}
    \bra{\psi}(\mathsf E_B+\mathsf E_C)\sfF
(\mathsf E_B+\mathsf E_C)\ket{\psi}
\leq2t\bra{\psi}\sfD\ket{\psi} + 2\bra{\psi}\sfE_B \sfF \sfE_C \ket{\psi} + 2\bra{\psi} \sfE_C \sfF \sfE_B \ket{\psi}.
\end{equation}
Chaining these together and using $|\bra{\psi}\sfE_B \sfF \sfE_C \ket{\psi}| \leq \norm{\sfE_B \sfF \sfE_C}_\infty$ implies:
$$4t^2 + 2t\bra{\psi}\sfD\ket{\psi} \leq 2t\bra{\psi}\sfD\ket{\psi} + 4\norm{\sfE_B \sfF \sfE_C}_\infty,$$
which implies $t^2 \leq \norm{\sfE_B \sfF(t) \sfE_C}_\infty$ as desired.
\end{proof}

\section{Statement on AI Use}

\begin{flushleft}
\renewcommand{\arraystretch}{1.3}
\begin{tabular}{@{}l@{\qquad}c@{}}
\textbf{Activity} & \textbf{AI used} \\
Asking the research question & No \\
Coming up with the approach & Yes \\
Proof development & Yes \\
Writing and exposition & Yes \\
Checking for bugs & Yes \\
Other supporting tasks & Yes \\
\end{tabular}
\end{flushleft}

\begin{sloppypar}
\paragraph{Details (elaborating on the points above).} 
The authors are fully accountable for the contents and correctness of this paper. 
The construction and proof were found by (Codex) GPT-5.6 Sol Ultra in concurrent and independent extended conversations with the authors, after which we decided to merge our original versions~\cite{ananth2026unconditionalunclonableencryption,ragavan2026efficientunclonableencryptionpauli}. This was in part using harness ideas generated by the authors based on the UCLA Moonshot Harness~\cite{MoonshotHarness} and~\cite{OpenAICDCPrompt}.
Preliminary versions of this paper were then drafted by GPT-5.6 Sol Ultra, which the human authors carefully verified and significantly cleaned up. 
GPT-6 Astra, GPT-5.6 Sol, and Claude Fable 5 were all used for proofreading and editorial tasks. 
Codex GPT-5.6 Sol and GPT-6 Astra were used to produce Lean 4 code for our proofs.
\end{sloppypar}

\section{Acknowledgements}
PA is supported by the National Science Foundation under the grants FET-2329938, CAREER-2341004, and FET-2530160.

SR thanks Alexander Poremba and Vinod Vaikuntanathan for introducing him to the problem of unclonable encryption and for a wonderful collaboration on this front.
SR additionally thanks Vinod Vaikuntanathan and Aparna Gupte for helpful feedback on drafts of this paper, and Alper Çakan, Eli Goldin, Yao-Ting Lin, and Justin Raizes for helpful discussions. 
SR's access to ChatGPT Pro was supported by the UK AISI alignment project and the Simons Institute for the Theory of Computing, and his access to Claude Max was supported by Jane Street.
SR's work on this was partly while a student researcher at Google Quantum AI; the remainder was supported by NSF CNS-2534400 and in part by Jane Street.

AS sincerely thanks Fatih Kaleoglu for patiently teaching him about unclonable encryption during a visit to UCSB long ago!

\newpage 

\printbibliography

@article{WoottersZurek82,
  author  = {Wootters, William K. and Zurek, Wojciech H.},
  title   = {A Single Quantum Cannot Be Cloned},
  journal = {Nature},
  volume  = {299},
  pages   = {802--803},
  year    = {1982},
  doi     = {10.1038/299802a0}
}

@online{OpenAICDCPrompt,
  author  = {{OpenAI}},
  title   = {Prompt Used for ``A Proof of the Cycle Double Cover Conjecture''},
  date    = {2026-07-09},
  url     = {https://cdn.openai.com/pdf/04d1d1e4-bc75-476a-97cf-49055cd98d31/cdc_prompt.pdf},
  urldate = {2026-07-22}
}

@misc{MoonshotHarness,
  author = {
    Junyi Zhang\textsuperscript{*} and
    Xinjie He\textsuperscript{*} and
    Hyunsik Chae and
    Ethan Ji and
    Eric Jiang and
    Rushil Raghavan and
    Yiwen Kou and
    Alex Taylor and
    Kai-Wei Chang\textsuperscript{\textdagger} and
    Raghu Meka\textsuperscript{\textdagger} and
    Violet Peng\textsuperscript{\textdagger} and
    Amit Sahai\textsuperscript{\textdagger} and
    Terence Tao\textsuperscript{\textdagger} and
    Wei Wang\textsuperscript{\textdagger}
  },
  title     = {UCLA Moonshot Harness},
  year      = {2026},
  shorthand = {ZHC+26},
  note      = {
    \textsuperscript{*}Co-first authors with equal contribution.
    The remaining students are ordered by contribution to the harness.
    \textsuperscript{\textdagger}Principal investigators, listed at the end
    in alphabetical order by last name.
  }
}

@misc{ananth2026unconditionalunclonableencryption,
      title={Unconditional Unclonable Encryption}, 
      author={Prabhanjan Ananth and Amit Sahai},
      year={2026},
      eprint={2607.21551},
      archivePrefix={arXiv},
      primaryClass={quant-ph},
      url={https://arxiv.org/abs/2607.21551}, 
}

@article{tomamichel2013monogamy,
  title={A monogamy-of-entanglement game with applications to device-independent quantum cryptography},
  author={Tomamichel, Marco and Fehr, Serge and Kaniewski, J{\k{e}}drzej and Wehner, Stephanie},
  journal={New Journal of Physics},
  volume={15},
  number={10},
  pages={103002},
  year={2013},
  publisher={IOP Publishing}
}

@misc{ragavan2026efficientunclonableencryptionpauli,
      title={Efficient Unclonable Encryption from Pauli Eigenstates}, 
      author={Seyoon Ragavan},
      year={2026},
      eprint={2607.21811},
      archivePrefix={arXiv},
      primaryClass={quant-ph},
      url={https://arxiv.org/abs/2607.21811}, 
}

@article{Dieks82,
  author  = {Dieks, Dennis},
  title   = {Communication by {EPR} Devices},
  journal = {Physics Letters A},
  volume  = {92},
  number  = {6},
  pages   = {271--272},
  year    = {1982},
  doi     = {10.1016/0375-9601(82)90084-6}
}

@article{Gottesman03,
  author  = {Gottesman, Daniel},
  title   = {Uncloneable Encryption},
  journal = {Quantum Information and Computation},
  volume  = {3},
  number  = {6},
  pages   = {581--602},
  year    = {2003},
  doi     = {10.26421/QIC3.6-2}
}

@inproceedings{DBLP:conf/stoc/MaH25,
  author       = {Fermi Ma and
                  Hsin{-}Yuan Huang},
  editor       = {Michal Kouck{\'{y}} and
                  Nikhil Bansal},
  title        = {How to Construct Random Unitaries},
  booktitle    = {Proceedings of the 57th Annual {ACM} Symposium on Theory of Computing,
                  {STOC} 2025, Prague, Czechia, June 23-27, 2025},
  pages        = {806--809},
  publisher    = {{ACM}},
  year         = {2025},
  url          = {https://doi.org/10.1145/3717823.3718254},
  doi          = {10.1145/3717823.3718254},
  bibsource    = {dblp computer science bibliography, https://dblp.org}
}

@misc{bartusek2026noteboostinguncloneableencryption,
      title={A Note on Boosting Uncloneable Encryption in {Microcrypt}}, 
      author={James Bartusek and Eli Goldin},
      year={2026},
      eprint={2605.27647},
      archivePrefix={arXiv},
      primaryClass={cs.CR},
      url={https://arxiv.org/abs/2605.27647}, 
}

@inproceedings{BroadbentLord20,
  author    = {Broadbent, Anne and Lord, S{\'e}bastien},
  title     = {Uncloneable Quantum Encryption via Oracles},
  booktitle = {15th Conference on the Theory of Quantum Computation,
               Communication and Cryptography (TQC 2020)},
  series    = {Leibniz International Proceedings in Informatics (LIPIcs)},
  volume    = {158},
  pages     = {4:1--4:22},
  publisher = {Schloss Dagstuhl--Leibniz-Zentrum f{\"u}r Informatik},
  year      = {2020},
  doi       = {10.4230/LIPIcs.TQC.2020.4}
}

@inproceedings{AnanthKaleoglu21,
  author    = {Ananth, Prabhanjan and Kaleoglu, Fatih},
  title     = {Unclonable Encryption, Revisited},
  booktitle = {Theory of Cryptography},
  series    = {Lecture Notes in Computer Science},
  volume    = {13042},
  pages     = {299--329},
  publisher = {Springer},
  year      = {2021},
  doi       = {10.1007/978-3-030-90459-3_11}
}

@inproceedings{AnanthEtAl22,
  author    = {Ananth, Prabhanjan and Kaleoglu, Fatih and Li, Xingjian
               and Liu, Qipeng and Zhandry, Mark},
  title     = {On the Feasibility of Unclonable Encryption, and More},
  booktitle = {Advances in Cryptology---CRYPTO 2022, Part II},
  series    = {Lecture Notes in Computer Science},
  volume    = {13508},
  pages     = {212--241},
  publisher = {Springer},
  year      = {2022},
  doi       = {10.1007/978-3-031-15979-4_8}
}

@inproceedings{AnanthKaleogluLiu23,
  author    = {Ananth, Prabhanjan and Kaleoglu, Fatih and Liu, Qipeng},
  title     = {Cloning Games: A General Framework for Unclonable Primitives},
  booktitle = {Advances in Cryptology---CRYPTO 2023, Part V},
  series    = {Lecture Notes in Computer Science},
  volume    = {14085},
  pages     = {66--98},
  publisher = {Springer},
  year      = {2023},
  doi       = {10.1007/978-3-031-38554-4_3}
}

@misc{MajenzSchaffnerTahmasbi21,
  author        = {Majenz, Christian and Schaffner, Christian
                   and Tahmasbi, Mehrdad},
  title         = {Limitations on Uncloneable Encryption and
                   Simultaneous One-Way-to-Hiding},
  year          = {2021},
  eprint        = {2103.14510},
  archivePrefix = {arXiv},
  primaryClass  = {quant-ph}
}

@inproceedings{AnanthKaleogluYuen25,
  author    = {Ananth, Prabhanjan and Kaleoglu, Fatih and Yuen, Henry},
  title     = {Simultaneous {Haar} Indistinguishability with Applications
               to Unclonable Cryptography},
  booktitle = {16th Innovations in Theoretical Computer Science Conference
               (ITCS 2025)},
  series    = {Leibniz International Proceedings in Informatics (LIPIcs)},
  volume    = {325},
  pages     = {7:1--7:23},
  publisher = {Schloss Dagstuhl--Leibniz-Zentrum f{\"u}r Informatik},
  year      = {2025},
  doi       = {10.4230/LIPIcs.ITCS.2025.7}
}

@inproceedings{ColadangeloLiuXie26,
  author       = {Andrea Coladangelo and
                  Qipeng Liu and
                  Ziyi Xie},
  editor       = {Shubhangi Saraf},
  title        = {The Curious Case of "XOR Repetition" of Monogamy-Of-Entanglement Games},
  booktitle    = {17th Innovations in Theoretical Computer Science Conference, {ITCS}
                  2026, Bocconi University, Milan, Italy, January 27-30, 2026},
  series       = {LIPIcs},
  volume       = {362},
  pages        = {41:1--41:20},
  publisher    = {Schloss Dagstuhl - Leibniz-Zentrum f{\"{u}}r Informatik},
  year         = {2026},
  url          = {https://doi.org/10.4230/LIPIcs.ITCS.2026.41},
  doi          = {10.4230/LIPICS.ITCS.2026.41},
  bibsource    = {dblp computer science bibliography, https://dblp.org}
}

@article{BhattacharyyaCulf26,
  author  = {Bhattacharyya, Archishna and Culf, Eric},
  title   = {Uncloneable Encryption from Decoupling},
  journal = {Nature Physics},
  volume  = {22},
  number  = {2},
  pages   = {315--318},
  year    = {2026},
  doi     = {10.1038/s41567-025-03154-7}
}

@article{navascues2008convergent,
  title={A convergent hierarchy of semidefinite programs characterizing the set of quantum correlations},
  author={Navascu{\'e}s, Miguel and Pironio, Stefano and Ac{\'\i}n, Antonio},
  journal={New Journal of Physics},
  volume={10},
  number={7},
  pages={073013},
  year={2008}
}

@article{broadbent2025optimaluntelegraphableencryptionimplications,
  author  = {Anne Broadbent and Eric Culf and Denis Rochette},
  title   = {Optimal Untelegraphable Encryption and Implications for Uncloneable Encryption},
  journal = {{IACR} Communications in Cryptology},
  volume  = {3},
  number  = {2},
  year    = {2026},
  doi     = {10.62056/akgy4fbmo}
}

@article{BotteronEtAl26,
  author  = {Botteron, Pierre and Broadbent, Anne and Culf, Eric
             and Nechita, Ion and Pellegrini, Cl{\'e}ment
             and Rochette, Denis},
  title   = {Towards Unconditional Uncloneable Encryption},
  journal = {Quantum},
  volume  = {10},
  pages   = {2157},
  year    = {2026},
  doi     = {10.22331/q-2026-07-08-2157}
}

@misc{BhattacharyyaBroadbentCulf26,
  author        = {Bhattacharyya, Archishna and Broadbent, Anne
                   and Culf, Eric},
  title         = {The Uncloneable Bit Exists},
  year          = {2026},
  eprint        = {2603.08916},
  archivePrefix = {arXiv},
  primaryClass  = {quant-ph}
}

@book{Watrous18,
  author    = {Watrous, John},
  title     = {The Theory of Quantum Information},
  publisher = {Cambridge University Press},
  year      = {2018},
  doi       = {10.1017/9781316848142}
}

@article{Jamiolkowski72,
  author  = {Jamio{\l}kowski, Andrzej},
  title   = {Linear Transformations Which Preserve Trace and Positive
             Semidefiniteness of Operators},
  journal = {Reports on Mathematical Physics},
  volume  = {3},
  number  = {4},
  pages   = {275--278},
  year    = {1972},
  doi     = {10.1016/0034-4877(72)90011-0}
}

@article{Choi75,
  author  = {Choi, Man-Duen},
  title   = {Completely Positive Linear Maps on Complex Matrices},
  journal = {Linear Algebra and its Applications},
  volume  = {10},
  number  = {3},
  pages   = {285--290},
  year    = {1975},
  doi     = {10.1016/0024-3795(75)90075-0}
}

@inproceedings{BartusekGoldin2026,
  author    = {James Bartusek and Eli Goldin},
  editor    = {Nadia Heninger and Mike Rosulek},
  title     = {Unclonable Encryption in the {Haar} Random Oracle Model},
  booktitle = {Advances in Cryptology -- {CRYPTO} 2026, 46th Annual International Cryptology Conference, Santa Barbara, CA, USA, August 17--20, 2026, Proceedings, Part V},
  series    = {Lecture Notes in Computer Science},
  pages     = {311--340},
  publisher = {Springer Nature Switzerland},
  year      = {2026},
  doi       = {10.1007/978-3-032-35409-9_11}
}

@inproceedings{ColadangeloGunn24,
  author = {Andrea Coladangelo and Sam Gunn},
  title = {How to Use Quantum Indistinguishability Obfuscation},
  booktitle = {Proceedings of the 56th Annual ACM Symposium on Theory of Computing},
  year = {2024},
  doi = {10.1145/3618260.3649779},
  url = {https://eprint.iacr.org/2023/1756}
}

@inproceedings{AnanthBehera24,
  author = {Prabhanjan Ananth and Amit Behera},
  title = {A Modular Approach to Unclonable Cryptography},
  booktitle = {Advances in Cryptology -- CRYPTO 2024},
  pages = {3--37},
  publisher = {Springer},
  year = {2024},
  doi = {10.1007/978-3-031-68394-7_1},
  url = {https://eprint.iacr.org/2023/1797}
}

@misc{BhattacharyyaBroadbentCulfMessages26,
  author = {Archishna Bhattacharyya and Anne Broadbent and Eric Culf},
  title = {Statistically Secure Uncloneable Encryption of Arbitrary Messages},
  year = {2026},
  eprint = {2607.28561},
  archivePrefix = {arXiv},
  primaryClass = {quant-ph},
  url = {https://arxiv.org/abs/2607.28561}
}

@misc{AnanthNonLocal26,
  author = {Prabhanjan Ananth},
  title = {Non-Local Search-to-Decision Reduction over {$\mathbb{F}_2$}, and More},
  year = {2026},
  eprint = {2608.19091},
  archivePrefix = {arXiv},
  primaryClass = {quant-ph},
  url = {https://arxiv.org/abs/2608.19091}
}

@misc{ColadangeloLiuXieGL26,
  author = {Andrea Coladangelo and Qipeng Liu and Ziyi Xie},
  title = {Unclonable Encryption from {BB84} States: A Simultaneous {Goldreich--Levin} Reduction},
  year = {2026},
  eprint = {2608.17629},
  archivePrefix = {arXiv},
  primaryClass = {quant-ph},
  url = {https://arxiv.org/abs/2608.17629}
}

@misc{BotteronDesignolleFawzi26,
  author = {Pierre Botteron and S{\'e}bastien Designolle and Omar Fawzi},
  title = {{Pauli} Encodings \& Unclonable Encryption},
  year = {2026},
  eprint = {2607.28602},
  archivePrefix = {arXiv},
  primaryClass = {quant-ph},
  url = {https://arxiv.org/abs/2607.28602}
}

\newpage

\appendix

\section{Glossary of Notation}
The table lists the notation used throughout the construction and
security proof. Auxiliary quantities used in a single proof are defined
where they appear.

{\small
\renewcommand{\arraystretch}{1.3}
\begin{longtable}{>{\raggedright\arraybackslash}p{0.19\textwidth}
                  >{\raggedright\arraybackslash}p{0.74\textwidth}}
\hline
\textbf{Symbol} & \textbf{Definition and Meaning}\\
\hline
\endfirsthead
\hline
\textbf{Symbol} & \textbf{Definition and Meaning}\\
\hline
\endhead
$\secparam,\ d$ & The ciphertext has $\secparam$ qubits and Hilbert-space
dimension $d=2^\secparam$.\\
$\regA,\ \regE,\ \regB,\ \regC$ & The reference register, ciphertext
register, and Bob's and Charlie's output registers, respectively.
The reference register is a copy of $\regE$.\\
$\cK_\secparam,\ K$ & The non-identity Pauli key space, of size
$K=4^\secparam-1=d^2-1$; see Definition~\ref{def:decoders}.\\
$k=(x,z),\ P_{x,z}$ & The secret key and its tensor Pauli
$P_{x,z}=\bigotimes_{j=1}^{\secparam}i^{x_jz_j}X^{x_j}Z^{z_j}$.\\
$[k,k']$ & For $k=(x,z)$ and $k'=(x',z')$, the Pauli symplectic form
$\langle x,z'\rangle+\langle z,x'\rangle\in\F_2$.\\
$\Xi^{\mathrm{Pauli}}$ & The $4^\secparam\times4^\secparam$ symplectic
character matrix with entries $(-1)^{[k,k']}$; see
Definition~\ref{def:symp}.\\
$m,\ b$ & The plaintext bit and challenge bit, respectively; they
coincide in the fixed-message one-bit game.\\
$q_j$ & Equals $1$ when the $j$th Pauli factor is non-identity, and $0$
otherwise. The preparation bits satisfy $\bigoplus_jq_jr_j=m$.\\
$U_k$ & The product of the local basis changes used to prepare the
ciphertext; see Section~\ref{sec:impl}.\\
$\rho_{\ct}^{(k,m)}$ & The ciphertext density operator
$d^{-1}(I+(-1)^mP_{x,z})$.\\
$\ket{\Omega}_{\regA\regE}$ & The normalized maximally entangled state
between the reference and ciphertext registers; see
Theorem~\ref{thm:basic:choi}.\\
$\Phi,\ \tau$ & The cloner's channel from $\regE$ to
$\regB\otimes\regC$, and its normalized Choi state on
$\regA\otimes\regB\otimes\regC$.\\
$\Theta_k$ & The transposed Pauli $(P_{x,z}^T)_{\regA}$ on the
reference register.\\
$B_k,\ C_k$ & Bob's and Charlie's Hermitian decoding observables for
key $k$, with norm at most one.\\
$\mathsf E_B,\ \mathsf E_C$ & The averaged reference--receiver
correlations defined in~\eqref{eqn:security:EB}
and~\eqref{eqn:security:EC}. Their expectations in $\tau$ are twice
the corresponding success probabilities minus one.\\
$\mathsf E_{BC}$ & The averaged receiver correlation
in~\eqref{eqn:security:EBC}. Its expectation is the probability of
agreement minus the probability of disagreement.\\
$\sfD$ & $I-\mathsf E_{BC}$. Its expectation is twice the probability
of disagreement.\\
$\bfG$ & $\tfrac12(\mathsf E_B+\mathsf E_C-\sfD)$, whose largest
eigenvalue controls the winning probability.\\
$p$ & The winning probability in the fixed-message one-bit game:
$p=\tfrac12(1+\Tr(\bfG\tau))$; see~\eqref{eq:endofpart1}.\\
$\gamma$ & $d/K$, the upper bound on
$\|\mathsf E_B\mathsf E_C\|_\infty$.\\
$t$ & A positive parameter in $\sfF(t)$; in the last step of the proof,
$t=\lambda_{\max}(\bfG)$ and the case $t\leq0$ is handled separately.\\
$\sfF(t)$ & $2t(2tI+\sfD)^{-1}$, defined for $t>0$.
Write $\sfF=\sfF(t)$ when $t$ is fixed.\\
\hline
\end{longtable}
}

\end{document}